\documentclass[acmsmall,nonacm]{acmart} 

\AtBeginDocument{%
  }

\setcopyright{acmcopyright}
\copyrightyear{2018}
\acmYear{2018}
\acmDOI{XXXXXXX.XXXXXXX}

\usepackage[utf8]{inputenc}
\usepackage{amsfonts,amsmath,amsthm}
\usepackage{enumerate}
\usepackage{mathrsfs}
\usepackage{comment}
\usepackage{graphicx} 
\usepackage{commath}
\usepackage{afterpage}
\usepackage{xcolor}
\usepackage{listings}
\usepackage{paralist}
\usepackage{enumitem}
\usepackage{caption}
\usepackage{subcaption}
\usepackage{hyperref}
\usepackage[english]{babel}
\usepackage{changepage}
\usepackage{thmtools,thm-restate}
\usepackage{mathtools}
\usepackage{multirow}
\usepackage{pbox}

\renewcommand{\paragraph}[1]{\smallskip\noindent\textbf{\emph{#1}}}

\begin{document}

\title{Equivalence and Similarity Refutation for Probabilistic Programs}

\newcommand{\PP}{{\sc PP}}
\newcommand{\support}{\mathit{supp}}
\newcommand{\dom}{\mathit{dom}}
\newcommand{\Vout}{V_{\mathit{out}}}
\newcommand{\pCFG}{\mathcal{C}}
\newcommand{\locinit}{\loc_{\mathit{init}}}
\newcommand{\locterm}{\loc_{\mathit{out}}}
\newcommand{\locfail}{\loc_{\mathit{fail}}}
\newcommand{\vecinit}{\mathbf{x}_{\mathit{init}}}
\newcommand{\lem}{\eta}
\newcommand{\preexp}[1]{\mathit{pre}_{#1}}
\newcommand{\varinit}{\Theta_{\mathit{init}}}
\newcommand{\updates}{\mathit{Up}}
\newcommand{\Fpath}{\mathit{Fpath}}
\newcommand{\Reach}{\mathit{Reach}}
\newcommand{\State}{\mathit{State}}
\newcommand{\Run}{\mathit{Run}}
\newcommand{\RunObs}{\mathit{Run}_{\texttt{observe}}}
\newcommand{\RunFail}{\mathit{Run}_{\texttt{fail}}}
\newcommand{\transitions}{\mapsto}
\newcommand{\probdist}{\mathit{Pr}}
\newcommand{\guards}{G}
\newcommand{\prob}{\mathit{Pr}}
\newcommand{\Inf}{\mathit{Inf}}
\newcommand{\Unif}{\mathit{Uniform}}
\newcommand{\Normal}{\mathit{Normal}}
\newcommand{\predicate}{\Psi}
\newcommand{\pvars}{V}
\newcommand{\rvars}{R}
\newcommand{\locs}{\mathit{L}}
\newcommand{\loc}{\ell}
\newcommand{\tran}{\tau}
\newcommand{\probm}{\mathbb{P}}
\newcommand{\E}{\mathbb{E}}
\newcommand{\expr}{E}
\newcommand{\val}{\mathbf{x}}
\newcommand{\todoP}[1]{\textbf{\textcolor{red}{TODO} P: #1}}
\newcommand{\todoDj}[1]{\textbf{\textcolor{blue}{TODO} Dj: #1}}
\newcommand{\Rset}{\mathbb{R}}
\newcommand{\sigmaAlg}{\mathcal{F}}
\newcommand{\TimeTerm}{\mathit{TimeTerm}}
\newcommand{\Termset}{\mathit{Term}}
\newcommand{\Output}{\mathit{Output}}
\newcommand{\TV}{\mathit{TV}}
\newcommand{\K}{\mathcal{K}}
\newcommand{\out}{\mathit{out}}
\newcommand{\Rsetnn}{\Rset_{\geq 0}}
\newcommand{\probmalt}{\mu}
\newcommand{\veca}[1]{\mathbf{#1}}
\newcommand{\nextv}{\mathit{Next}}
\newcommand{\stime}{T}
\newcommand{\indicator}[1]{\mathbb{I}_{#1}}
\newcommand{\Nset}{\mathbb{N}}
\newcommand{\Mono}{\textit{Mono}}
\newcommand{\pcfg}{\pCFG}

\newtheorem{problem}{Problem}
\newtheorem{remark}{Remark}

\author{Krishnendu Chatterjee}
\email{krishnendu.chatterjee@ist.ac.at}
\orcid{0000-0002-4561-241X}
\affiliation{
    \institution{Institute of Science and Technology Austria (ISTA)}
    \city{Klosterneuburg}
    \country{Austria}
}

\author{Ehsan Kafshdar Goharshady}
\email{ehsan.goharshady@ist.ac.at}
\orcid{0000-0002-8595-0587}
\affiliation{
    \institution{Institute of Science and Technology Austria (ISTA)}
    \city{Klosterneuburg}
    \country{Austria}
}

\author{Petr Novotný}
\email{petr.novotny@fi.muni.cz}
\orcid{0000-0002-5026-4392}
\affiliation{
    \institution{Masaryk University}
    \city{Brno}
    \country{Czech Republic}
}

\author{\DJ or\dj e \v{Z}ikeli\'c}\authornote{Part of the work done while the author was at the Institute of Science and Technology Austria (ISTA).}
\email{dzikelic@smu.edu.sg}
\orcid{0000-0002-4681-1699}
\affiliation{
    \institution{Singapore Management University}
    \country{Singapore}
}


\begin{abstract}
	We consider the problems of statically refuting equivalence and similarity of output distributions defined by a pair of probabilistic programs. Equivalence and similarity are two fundamental relational properties of probabilistic programs that are essential for their correctness both in implementation and in compilation. In this work, we present a new method for static equivalence and similarity refutation. Our method refutes equivalence and similarity by computing a function over program outputs whose expected value with respect to the output distributions of two programs is different. The function is computed simultaneously with an upper expectation supermartingale and a lower expectation submartingale for the two programs, which we show to together provide a formal certificate for refuting equivalence and similarity. To the best of our knowledge, our method is the first approach to relational program analysis to offer the combination of the following desirable features: (1)~it is fully automated, (2)~it is applicable to infinite-state probabilistic programs, and (3)~it provides formal guarantees on the correctness of its results. We implement a prototype of our method and our experiments demonstrate the effectiveness of our method to refute equivalence and similarity for a number of examples collected from the literature.
\end{abstract}




\maketitle

\section{Introduction}\label{sec:intro}

\noindent{\bf\em Probabilistic programs.} Probabilistic programs are imperative or functional programs extended with the ability to perform sampling from probability distributions and to condition data on observations~\cite{GordonHNR14,MeentPYW18,barthe2020foundations}. They provide an expressive framework for specifying probabilistic models and have been adopted in a range of application domains including stochastic networks~\cite{FosterKMR016}, machine learning~\cite{Ghahramani15}, security~\cite{BartheGGHS16,BartheGHP16} and robotics~\cite{Thrun00}. Instead of designing different inference and analysis techniques for probabilistic models that may arise in each of these domains, one can first specify the probabilistic model of interest as a probabilistic program and then utilize the existing techniques for probabilistic programs. This separation of model specification on one hand and inference and analysis on the other hand has sparked interest in the probabilistic programming paradigm, and recent years have seen the development of many probabilistic programming languages, e.g.~Church~\cite{GoodmanMRBT08}, Pyro~\cite{BinghamCJOPKSSH19} or Edward~\cite{tran2017deep}. Concurrently with studying the design and implementation of probabilistic programming languages, formal analysis of probabilistic programs has also become a very active research area.

\smallskip\noindent{\bf\em Static analysis of probabilistic programs.} Probabilistic programs are hard to reason about. While deterministic programs always produce the same output on a given input, probabilistic programs give rise to {\em output distributions}. This makes probabilistic programs extremely hard to analyze both in theory~\cite{KaminskiKM19} and in practice~\cite{NandiGSMM17,DuttaLHM18}, as bugs in probabilistic program implementation may be very subtle and hard to detect.

Recent years have seen much work on static analysis of probabilistic programs, where the aim is to formally prove temporal or input/output properties by analyzing the source code directly and instead of repeatedly sampling randomized executions of probabilistic programs. There have been significant developments on static analysis with respect to termination~\cite{ChakarovS13,ChatterjeeFNH18,KaminskiKMO18,McIverMKK18,ChatterjeeGMZ22,AgrawalC018,ChatterjeeGNZZ21}, reachability~\cite{TakisakaOUH21}, safety~\cite{ChatterjeeNZ17,SankaranarayananCG13,BeutnerOZ22,BatzCJKKM23,BaoTPHR22}, cost~\cite{NgoC018,Wang0GCQS19,Wang0R21}, input/output~\cite{ChenKKW22}, runtime~\cite{DBLP:conf/tacas/KuraUH19}, productivity for infinite streams~\cite{0001BH018}, sensitivity~\cite{BartheEGHS18,WangFCDX20,0001BHKKM21} or differential privacy~\cite{AlbarghouthiH18} properties. 

\smallskip\noindent{\bf\em Equivalence and similarity refutation.} In this work, we focus on \emph{relational analysis} of probabilistic programs. The goal of relational analysis is to prove properties of \emph{pairs} of probabilistic programs. A prominent example of relational property is \emph{equivalence}: two probabilistic programs are equivalent if they define the same output distributions. In this paper, we consider static analysis of equivalence and similarity of output distributions of probabilistic program pairs. Equivalence and similarity are two fundamental properties of probabilistic programming systems that are essential for their correctness {\em both in implementation and in compilation}. We study the following two problems:
\begin{compactenum}
	\item {\em Equivalence refutation problem.} Given a pair of probabilistic programs, prove that their output distributions are not equivalent (a notion formally defined in Section~\ref{sec:problem}).
	\item {\em Similarity refutation problem.} Given a pair of probabilistic programs, prove a lower bound on Kantorovich distance~\cite{villani2021topics} between their output distributions (we formally define Kantorovich distance and discuss its relation to other distances in Sections~\ref{sec:distances} and~\ref{sec:problem}).
\end{compactenum}

\smallskip\noindent{\bf\em Relevance.} Equivalence checking and refutation are crucial for ensuring probabilistic program correctness or for bug detection. For instance, if we have two different implementations of a probabilistic model or two randomized algorithms designed to solve a given problem, the equivalence refutation analysis allows us to detect whether the two probabilistic programs give rise to different output distributions~\cite{MurawskiO05}. Such an analysis allows, e.g., bug detection in samplers from probability distributions~\cite{ChakrabortyM19} or in implementations of cryptographic protocols~\cite{BartheGB09}. Equivalence refutation analysis also allows bug detection in probabilistic program compilers. For instance, it was observed by~\cite{0001ZHM19} that a 10-line probabilistic program in Stan~\cite{gelman2015stan} executes over 6000 lines of code of Stan implementation. Hence, detecting compilation bugs by testing may be a challenging task even for small programs. Static equivalence refutation analysis allows bug detection in compilation by comparing the source code to its intermediate representation without program execution.

While equivalence refutation analysis only proves that output distributions of two programs are not equivalent, similarity analysis provides more fine-grained information and {\em quantifies the difference} between two output distributions (e.g., the difference between the output distribution induced by a sampler and the ground probability distribution whose samples we wish to generate~\cite{ChakrabortyM19}).

\smallskip\noindent{\bf\em Prior work.} 
Equivalence and similarity are {\em relational properties} that are defined with respect to a program {\em pair}. The prior work on relational reasoning about probabilistic programs focused on developing logical systems for such reasoning~\cite{culpepper2017contextual} rather than on automation; or on sensitivity analysis~\cite{BartheEGHS18,0001BHKKM21,WangFCDX20,HuangWM18,AlbarghouthiH18,BartheGGHS16}, whose aims and assumptions differ from equivalence analysis (see Section~\ref{sec:relatedwork} for detailed discussion). \emph{Automated} methods for formal analysis of probabilistic program equivalence have been developed for \emph{finite-state} probabilistic systems~\cite{MurawskiO05,KieferMOWW11,BartheJK22}. However, probabilistic programs defined over real or integer-valued variables or containing sampling from continuous probability distributions (such as normal or uniform) all give rise to infinite-state programs. 

On the other hand, there is a huge body of work on sampling-based statistical testing of equivalence and similarity of two probability distributions~\cite{BatuFRSW13,ChanDVV14}, see~\cite{canonne2020survey} for a survey. While these methods provide extremely useful information and do not impose syntactic restrictions on probability distributions that they can analyze, they suffer from two key limitations. First, guarantees on the correctness of their analyses are {\em statistical}, meaning that there is always a non-zero probability that the analysis results are incorrect. Second, sampling-based methods suffer from scalability issues if the probabilistic program needs to be executed for a long time. For instance, the two programs in Figure~\ref{fig:motivatingtwo} both consist of 7 lines of code; however, each execution of either of the two programs requires millions of samples from uniform distribution.
Static analysis methods would be much more appropriate for analyzing equivalence or similarity of such programs.

To the best of our knowledge, no prior work has proposed an {\em automated} method for equivalence and similarity refutation analyses in {\em infinite-state probabilistic programs} that provide {\em formal guarantees} on the correctness of their results.

\smallskip\noindent{\bf\em Our approach -- automated formal analysis via expectation martingales.} We present a new method for static equivalence and similarity refutation analyses of probabilistic program pairs. To the best of our knowledge, we present the first method that provides the following desired features:
\begin{compactenum}
	\item {\em Automation.} Our method is fully automated.
	\item {\em Infinite-state programs.} Our method is applicable to infinite-state probabilistic programs.
	\item {\em Formal guarantees.} Our method provides formal guarantees on the correctness of its results.
\end{compactenum}
\smallskip\noindent{\bf\em Technical challenges.} Given two programs, our method refutes their equivalence by computing a function \( f \) over their output variables such that the expected value of \( f \) at the output of the two programs differs. Our method searches for such a function by computing it simultaneously with an {\em upper expectation supermartingale (UESM)} for the first program and a {\em lower expectation submartingale (LESM)} for the second program. UESMs and LESMs, notions similar to cost supermartingales~\cite{Wang0GCQS19} or super- and sub-invariants~\cite{HarkKGK20} (see Remark~\ref{rmk:comparison} for a comparison), provide sound proof rules for deriving upper and lower bounds on the expected value of a function on program output in probabilistic programs. We show that UESMs and LESMs together with the function $f$ over outputs provide sound proof rules for refuting equivalence and similarity of programs. To the best of our knowledge, no martingale-based approach has been used in prior work for static analysis of {\em relational} properties of probabilistic program pairs. 
The non-trivial challenge is to simultaneously compute the function $f$, along with two martingales (one submartingale and other supermatingale), which we achieve via a constraint solving-based approach.


\smallskip\noindent{\bf\em Contributions.} Our contributions can be summarized as follows:
\begin{compactenum}
	\item To our best knowledge, we present the first method for {\em static equivalence and similarity refutation} of probabilistic program pairs, which is {\em automated}, applicable to {\em infinite-state} probabilistic programs and provides {\em formal guarantees} on the correctness of its results.
	
	
	\item We formulate {\em sound proof rules for equivalence and similarity refutation} via UESMs and~LESMs.
	
	\item We present {\em fully automated algorithms} for equivalence and similarity refutation analyses in probabilistic programs, based on the above proof rules. 
	The algorithms simultaneously compute a UESM and an LESM for two probabilistic programs together with a function over their output variables. 
	They are applicable to numerical probabilistic programs with polynomial arithmetic expressions that may contain sampling instructions from both discrete and continuous probability distributions.
    Moreover, our method and our algorithm for similarity refutation are also applicable to other distance metrics, such as Total Variation~(TV), which can be reduced to the Kantorovich distance (see Section~\ref{sec:distances} for details).
	
	\item Our {\em experimental evaluation} demonstrates the ability of our method to refute equivalence and compute lower bounds on Kantorovich distance for a variety of program pairs.
\end{compactenum}



\begin{figure}[t]
\centering
\begin{subfigure}{0.47\textwidth}
\begin{lstlisting}[frame=none,numbers=none,escapechar=@,mathescape=true]
   $\,\texttt{sent} = 0,\, \texttt{fail} = 0$
$\locinit$:while $\texttt{sent} \leq 8\,000\,000$ and $\texttt{fail} \leq 0$:
$\loc_1$:      if prob($0.999$):
$\loc_2$:          $\texttt{sent} = \texttt{sent} + 1$
$\loc_3$:      else:
$\loc_4$:          $\texttt{fail} = 1$
$\locterm$:return $\texttt{sent}$
\end{lstlisting}
\end{subfigure}
\hfill
\begin{subfigure}{0.49\textwidth}
\begin{lstlisting}[frame=none,numbers=none,escapechar=@,mathescape=true]
   $\,\texttt{sent} = 0,\, \texttt{fail} = 0$
$\locinit$:while $\texttt{sent} \leq 9\,000\,000$ and $\texttt{fail} \leq 0$:
$\loc_1$:      if prob($0.9995$):
$\loc_2$:          $\texttt{sent} = \texttt{sent} + 1$
$\loc_3$:      else:
$\loc_4$:          $\texttt{fail} = 1$
$\locterm$:return $\texttt{sent}$
\end{lstlisting}
\end{subfigure}
\vspace{-0.5em}
\caption{Transmission protocol example.}
\vspace{-1em}
\label{fig:motivatingtwo}
\end{figure}

\section{Overview}\label{sec:overview}

We start by presenting an overview of our approach and illustrating it on the probabilistic program pair in Figure~\ref{fig:motivatingtwo}. We first overview our method for solving the equivalence refutation problem, and then show how our method can be extended to also solve the similarity refutation problem. We provide two more motivating examples for the equivalence and the similarity refutation problems in Section~\ref{sec:further-motive} in the Supplementary material. 

\begin{example}[Simple programs with long execution times]
	Consider the probabilistic program pair in Figure~\ref{fig:motivatingtwo}. Each program models a simplified network protocol~\cite{HelminkSV93,BatzCJKKM23} which aims to transmit $\texttt{n}$ packets from the receiver to the sender. However, each packet may be lost with probability $\texttt{p}$, and the transmission stops whenever some packet is lost. For the program in Figure~\ref{fig:motivatingtwo} left, we have $\texttt{n} = 8\,000\,000$ and $\texttt{p} = 0.001$ as in~\cite{BatzCJKKM23}. On the other hand, the protocol in Figure~\ref{fig:motivatingtwo} right transmits $\texttt{n} = 9\,000\,000$ packets with loss probability $\texttt{p} = 0.0005$. Both programs output the number $\texttt{sent}$ of successfully transmitted packets, hence the output distribution of each program is the probability distribution of the value of $\texttt{sent}$ upon termination.
	
	One easily sees that these two programs do not define equivalent output distributions. However, using sampling-based statistical testing to deduce this would be extremely inefficient. Indeed, sampling a single execution of either program requires $\texttt{n} = 8\,000\,000$  or $\texttt{n} = 9\,000\,000$ samples from Bernoulli distribution, respectively. A static analysis approach that does not need to sample program executions would be much more appropriate for refuting equivalence in this example.
\end{example}

\noindent{\bf\em Requirements.} In the sequel, we consider a pair of probabilistic programs and assume that they satisfy the following requirements. First, we assume that the programs share a common set of {\em output variables $\Vout$}. This is necessary for the output distributions to be defined over the same space so that they can be compared. Second, we assume that both programs are {\em almost-surely terminating}, so that their output distributions are indeed probability distributions.

\smallskip\noindent{\bf\em Equivalence refutation.} Let $\mathbb{E}_{\mu_1}$ and $\mathbb{E}_{\mu_2}$ denote the expectation operators over output distributions defined by two probabilistic programs. Our method refutes equivalence by searching for a function $f: \mathbb{R}^{|\Vout|} \rightarrow \mathbb{R}$ that maps program outputs to real numbers, whose expected values over two output distributions are not equal, i.e.~$\mathbb{E}_{\mu_1}[f] \neq \mathbb{E}_{\mu_2}[f]$.

To find such a function $f$, our method simultaneously searches for an {\em upper expectation supermartingale (UESM)} for the first program and a {\em lower expectation submartingale (LESM)} for the second program, notions that we formally define in Section~\ref{sec:expectationsupermartingales}. 
For a probabilistic program and a function $f$ over its outputs, a UESM for $f$ (resp.~LESM for $f$) provides a sound proof rule for deriving an upper bound (resp.~lower bound) of the expected value of $f$ on program output. Hence, in order to refute equivalence, our method searches for
\begin{compactenum}
	\item a function $f: \mathbb{R}^{|\Vout|} \rightarrow \mathbb{R}$ over program outputs,
	\item an UESM for $f$ in the first program, and
	\item an LESM for $f$ in the second program,
\end{compactenum}
such that the upper bound on $\mathbb{E}_{\mu_1}[f]$ implied by the UESM is strictly smaller than the lower bound on $\mathbb{E}_{\mu_2}[f]$ by by the LESM. In Section~\ref{sec:expectationsupermartingales}, we show that these three objects together formally certify that $\mathbb{E}_{\mu_1}[f] \neq \mathbb{E}_{\mu_2}[f]$ and thus that the output distributions of two programs are not equivalent.

Note that searching for a function $f$ over outputs whose expectation differs in the two programs yields {\em both sound and complete} proof rule for refuting equivalence of output distributions. Indeed, we will prove soundness in Section~\ref{sec:expectationsupermartingales} as stated above. On the other hand, for completeness, suppose that two output distributions are not equivalent. Then, there exists an event $A$ over outputs such that $P_{\mu_1}[A]\neq P_{\mu_2}[A]$. Hence, with $f$ being the indicator function $I(A)$, we have $E_{\mu_1}[I(A)]\neq E_{\mu_2}[I(A)]$.

\smallskip\noindent{\bf\em Upper and lower expectation martingales.} Consider a probabilistic program and a function $f$ over its outputs. Intuitively, an {\em upper expectation supermartingale (UESM)} for $f$ is a function $U_f$ that assigns a real value to each program state (comprising of a location in the code and program variable values), which is required to satisfy the following two conditions:
\begin{compactenum}
	\item {\bf Zero on output} The function $U_f$ is equal to zero on termination, i.e.~$U_f(s)=0$ for every reachable terminal state $s$ in the program.
	\item {\bf Expected $f$-decrease} In every step of program computation, an increase in the value of $f$ is matched in expectation by the decrease in the value of $U_f$. That is, for every reachable state $s$ in the program, we have $U_f(s) - \mathbb{E}[U_f(s')] \geq \mathbb{E}[f((s')^{\text{out}})] - f(s^{\text{out}})$.
\end{compactenum}
Here, we use the standard primed notation from program analysis: $s'$ denotes the probabilistically chosen successor of the state $s$ upon one computational step of the program. Also, $s^{\text{out}}$ and $(s')^{\text{out}}$ denote the output variable valuations defined by states $s$ and $s'$.

The expected \(f\)-decrease condition can be rewritten as $U_f(s) \geq \mathbb{E}[U_f(s')] + \mathbb{E}[f((s')^{\text{out}})] - f(s^{\text{out}})$.
Intuitively, \(U_f(s)\) is an upper bound on the expected difference between the value of \(f\) \emph{in the current state} (which is \(f(s^{\text{out}})\)) and upon termination (which is a random variable over the output distribution of paths starting from  \(s\)).

Lower expectation submartingales (LESMs) are defined similarly, with the only difference being that the expected $f$-decrease is replaced by the dual \( f \)-increase (by replacing \( \geq \) with \( \leq \)). 

We formally define UESMs and LESMs in Section~\ref{sec:expectationsupermartingales}. Furthermore, we prove that a UESM in the initial state of the program evaluates to an {\em upper bound} on the difference between the expected value of $f$ on output and the value of $f$ in the initial program state (subject to at least one of the so-called ``Optional Stopping Theorem'' conditions being satisfied, see Section~\ref{sec:expectationsupermartingales} for details); and dually for LESMs. Hence, U/LESMs provide a sound proof rule for computing upper/lower bounds on the expected value of a function defined over program outputs.
%
%
The names of UESMs and LESMs emphasize their connection to supermartingale and submartingale processes in probability theory, respectively~\cite{Williams91}, which lie at the core of soundness proofs of our proof rules. Intuitively, supermartingales (resp.~submartingales) are a class of stochastic processes that decrease (resp.~increase) in expected value upon every one-step evolution of the process. In particular, in the case of UESMs, we see from the above definition that the sum of $U_f$ and $f$ intuitively behaves like a supermartingale, and similarly for LESMs and submartingales. 

\begin{example}\label{ex:ulesm}
Consider the programs shown in Figure~\ref{fig:motivatingtwo} with output variables $\texttt{sent}$ and $\texttt{fail}$. Define the function $f(\texttt{sent},\texttt{fail}) = \texttt{sent} - \text{fail}$ over the outputs of programs. Furthermore, define the functions $U_f$ mapping states in the left program to reals and $L_f$ mapping states in the right program to reals via (as computed by our tool in Section~\ref{sec:exper}, rounded to one decimal)
\begin{equation*}
	U_f \begin{pmatrix}
		\loc,\\
		\texttt{sent},\\
		\texttt{fail}
	\end{pmatrix}
	= \begin{cases}
		998 - 998 \cdot \texttt{fail}, &\text{if } \loc = \locinit \\
		998 - 997 \cdot \texttt{fail}, &\text{if } \loc = \loc_1 \\
		999 - 998 \cdot \texttt{fail}, &\text{if } \loc = \loc_2 \\
		-1 +\texttt{fail}, &\text{if } \loc = \loc_3 \\
		-1 + \texttt{fail} , &\text{if } \loc = \loc_4 \\
		0, &\text{if } \loc = \locterm \\
	\end{cases}
	\quad
	L_f\begin{pmatrix}
		\loc,\\
		\texttt{sent},\\
		\texttt{fail}
	\end{pmatrix}
	= \begin{cases}
		1997.5 - 1997.5 \cdot \texttt{fail}, &\text{if } \loc = \locinit \\
		1997.5 - 1996.5\cdot\texttt{fail}, &\text{if } \loc = \loc_1 \\
		1998.5 - 1997.5 \cdot\texttt{fail}, &\text{if } \loc = \loc_2 \\
		-1 + \texttt{fail}, &\text{if } \loc = \loc_3 \\
		-1 + \texttt{fail}, &\text{if } \loc = \loc_4 \\
		0, &\text{if } \loc = \locterm \\
	\end{cases}
\end{equation*}
Since both functions are equal to $0$ at all reachable output states, it follows that they both satisfy the Zero on output condition. Furthermore, one can check by inspection that $U_f$ satisfies the Expected $f$-decrease condition in the program on the left, and that $L_f$ satisfies the Expected $f$-increase condition in the program on the right. Hence, $U_f$ is an example of an UESM for $f$ in the program in the left, and $L_f$ is an example of an LESM for $f$ in the program in the right.
\end{example}

\noindent{\bf\em UESMs and LESMs for equivalence refutation.} To refute equivalence of two probabilistic programs, our method computes (1)~a function $f$ over probabilistic program outputs, (2)~an UESM $U^1_f$ for $f$ in the first program, and (3)~an LESM $L^2_f$ for $f$ in the second program, such that
\[ U_f(s^1_{\text{init}}) + f((s^1_{\text{init}})^{\text{out}}) < L_f(s^2_{\text{init}}) + f((s^2_{\text{init}})^{\text{out}}), \] 
where $s^1_{\text{init}}$ and $s^2_{\text{init}}$ are the initial states of the first and the second program, respectively. Note that the choice of computing UESMs for the first program and LESMs for the second program rather than the opposite is made without loss of generality. Indeed, by simply negating the function $f$, an UESM for $f$ becomes an LESM for $-f$ and vice-versa. We formalize our proof rule for the equivalence refutation problem and prove its soundness in Section~\ref{sec:proofequivalencesimilarity}.

\begin{example}\label{ex:uelsmequiv}
	Consider again the programs in Figure~\ref{fig:motivatingtwo} and the function $f$, 
	the UESM $U_f$, 
	and the LESM $L_f$ 
	defined in Example~\ref{ex:ulesm}. The initial state of the programs satisfies $\texttt{sent} = \texttt{fail} = 0$. Hence,
	\[ U_f(s^1_{\text{init}}) + f((s^1_{\text{init}})^{\text{out}}) = 998 < 1997.5 = L_f(s^2_{\text{init}}) + f((s^2_{\text{init}})^{\text{out}}). \]
	Hence, our method refutes equivalence of output distributions of programs in Figure~\ref{fig:motivatingtwo}.
\end{example}

\noindent{\bf\em Automation: Simultaneous synthesis.} The key challenge in automating the aforementioned idea is the effective computation of the function over outputs, the UESM and the LESM. Note that these objects cannot be computed separately -- the computation must be guided by the objective of obtaining $f$, $U^1_f$ and $L^2_f$ such that $U_f(s^1_{\text{init}}) + f((s^1_{\text{init}})^{\text{out}}) < L_f(s^2_{\text{init}}) + f((s^2_{\text{init}})^{\text{out}})$.

We solve this challenge by employing a constraint-solving-based approach to compute these three objects {\em simultaneously}. While our theoretical results apply to the general arithmetic probabilistic programs, our automated method is applicable to probabilistic program pairs in which all arithmetic expressions are polynomials over program variables. It first fixes a polynomial template for $f$ by fixing a symbolic polynomial expression over output variables $\Vout$. It also fixes polynomial templates for the UESM in the first program and the LESM in the second program by fixing one symbolic polynomial expression over program variables at each location of each program. The defining conditions of the UESM
and the LESM 
are then encoded as constraints over the symbolic template variables. In addition, we add the {\em equivalence refutation constraint} $U_f(s^1_{\text{init}}) + f((s^1_{\text{init}})^{\text{out}}) < L_f(s^2_{\text{init}}) + f((s^2_{\text{init}})^{\text{out}}$. This results in a system of constraints whose every solution gives rise to a concrete instance of $f$, $U^1_f$ and $L^2_f$ that refute equivalence. Our synthesis then proceeds by solving the resulting system of constraints.

Note that considering $f$, $U^1_f$ and $L^2_f$ specified in terms of polynomials over program variables allows us to capture both expectations as well as higher moments of any random variable defined in terms of a polynomial expression over program variables in the output probability space of each program. We present our algorithm in Section~\ref{sec:algo}.



\paragraph{Extension to similarity refutation.} Our method for solving the equivalence refutation problem can be adapted to the similarity refutation problem. In particular, if we additionally require that the function $f$  is $1$-Lipschitz continuous, we show that
\[ L_f(s^2_{\text{init}}) + f((s^2_{\text{init}})^{\text{out}}) - U_f(s^1_{\text{init}}) - f((s^1_{\text{init}})^{\text{out}}) \]
evaluates to a lower bound on the Kantorovich distance between the output distributions of the two programs. We omit the details in order to keep this overview non-technical. We define Kantorovich distance and Lipschitz continuity in Section~\ref{sec:distances}, prove the soundness of UESMs and LESMs for proving lower bounds on Kantorovich distance in Section~\ref{sec:proofequivalencesimilarity} and show how to impose the additional $1$-Lipschitz continuity condition in our automated synthesis procedure in Section~\ref{sec:algo}.

\begin{example}\label{ex:uelsmsimil}
	Going back to the probabilistic program pair in Figure~\ref{fig:motivatingtwo} and Examples~\ref{ex:ulesm} and~\ref{ex:uelsmequiv}, since the function $f(\texttt{sent},\texttt{fail}) = \texttt{sent} - \texttt{fail}$ is $1$-Lipschitz with respect to the $L^1$-distance over $\mathbb{R}^2$, it immediately follows from our result in Example~\ref{ex:uelsmequiv} that the Kantorovich distance between output distributions of these programs is bounded from below by
	\[ L_f(s^2_{\text{init}}) + f((s^2_{\text{init}})^{\text{out}}) - U_f(s^1_{\text{init}}) + f((s^1_{\text{init}})^{\text{out}}) = 1997.5 - 998 = 999.5. \]
\end{example}

\paragraph{Limitations.} While our experimental results demonstrate the applicability of our method to a wide range of probabilistic program pairs, our approach has several limitations:
\begin{compactenum}
    \item Currently, our approach does not support programs with conditioning. 
    \item In general, the lower bounds on the Kantorovich distance of output distributions computed by our approach might not be tight.
    \item From a practical perspective, the performance of our automated method is dependent on the quality of {\em supporting linear invariants} generated for both programs. In our approach, these are computed by off-the-shelf invariant generators. See Section~\ref{sec:exper} for details.
\end{compactenum}

\begin{remark}[Martingale-based approach to relational analysis]\label{rmk:comparison}
	Martingale-based approach has been widely studied for static analysis of probabilistic programs, and UESMs and LESMs used in our approach are based on cost supermartingales~\cite{Wang0GCQS19} or super- and sub-invariants for expectation bounds~\cite{HarkKGK20} in single programs. In contrast to these concepts, our key differences are: (a) we consider proof rules for {\em relational analysis} of equivalence and similarity properties of program pairs; (b) we consider proof rules based on {\em both} super- {\em and} submartingales; (c) we consider the two types of martingales (UESMs and LESMs) {\em simultaneously}; and (d) in addition to synthesizing a supermartingale and a submartingale, we also need to simultaneously synthesize a function $f$ on outputs with respect to which the UESM and the LESM are defined.
	
	Furthermore, our results on UESMs and LESMs subsume and unify the results of~\cite{Wang0GCQS19,HarkKGK20}. Moreover, while~\cite{HarkKGK20} make the assumption of non-negative program variables and leave the generalization to programs with both positive and negative variables as a direction of future work~\cite[page~26]{HarkKGK20}, our UESM/LESMs apply to both positive and negative variables under the same assumptions as in~\cite{HarkKGK20}. The non-negative variables assumption is also imposed by the methods~\cite{AvanziniMS20,Wang0R21} for automated computation of bounds on expected values, whereas our UESM/LESMs are applicable to programs with both positive and negative variables. More detailed discussion of the differences is provided in Section~\ref{sec:proofrulesbounds}.
\end{remark}

\section{Preliminaries}\label{sec:prelims}


We use boldface notation for vectors, e.g. \( \mathbf{x}, \mathbf{y}, \) etc. An \( i \)-component of vector \( \val \) is denoted by \( \val[i] \). For an \( n \)-dimensional vector \( \val \), index \( 1 \leq i \leq n \), and number \( a \) we denote by \( \val(i\leftarrow a) \) the vector \( 
\mathbf{y} \) s.t. \( \mathbf{y}[i] = a \) and \( \mathbf{y}[j]=\val[j] \) for all \( 1 \leq j \leq n \) s.t. \( j \neq i \). Throughout the paper, we work with vectors representing valuations of variables of some program. We assume some canonical ordering of the variables, denoting them \( x_1, x_2, x_3,\ldots \), though in our examples we use aliases \( x, y, z, \ldots \) for better readability. Hence, for a program with \( n \) variables \( x_1,\ldots, x_n \), the number \( \val[i] \) denotes the value of variable \( x_i \) in valuation \( \val \in \Rset^n \).


We will operate with some basic notions of probability theory, such as \emph{probability space}, \emph{random variable}, \emph{expected value,} etc. We review the formal definitions of these notions in Section~\ref{app:probt} of the Supplementary material. We use the term \emph{probability distribution} interchangeably with \emph{probability measure}, particularly when the underlying sample space is (some subset of) an Euclidean space \( \Rset^n \). For a finite or countable set \( A \), we denote by \( \mathcal{P}(A) \) the set of all probability distributions on \( A \).

\vspace{-0.5em}
\subsection{Program Syntax}\label{sec:prelimssyntax}

\noindent{\bf\em Imperative-style syntax.} We consider imperative arithmetic programs consisting of standard programming constructs: variable assignments, sequential composition, conditional branching, and loops. Right-hand sides of variable assignments are formed by expressions built from constants, program variables and Borel-measurable arithmetic operators (Borel measurability~\cite{Williams91} is a standard assumption in probabilistic program analysis that is satisfied by all standard arithmetical operators). We denote by \( \expr(\val) \) the value of expression \( E \) in valuation \( \val \) and assume \( \expr(\val) \) to be well-defined for all valuations \( \val \). The guards of loops and conditional statements consist of \emph{predicates.} A predicate \( \Psi \) is a logical formula obtained by a finite number of applications of conjunction, disjunction, and negation operations on \emph{atomic predicates} of the form \( E_1 \leq E_2 \), where \( E_1, E_2 \) are expressions. We denote by \( \val \models \Psi \) the fact that a predicate \( \Psi \) is satisfied by the valuation~\( \val \).

%

\smallskip\noindent{\bf\em Probabilistic instructions.} Our programs also admit two types of \emph{probabilistic} statements. The first is \emph{probabilistic branching,} in our examples represented by the command \(\textbf{if prob}(p)\textbf{ then ... else ...}\). Upon the execution of such a statement, the program enters the if-branch with probability \( p \) and the else-branch with probability \( 1-p \). The second is \emph{sampling} of a variable value from a given probability distribution, represented by the \( \textbf{sample(\ldots)} \) statement in our examples. We allow sampling from both discrete and continuous probability distributions. In this work, we do not consider conditioning on observations.
	

Figure~\ref{fig:motivatingtwo} shows the typical form of the programs we work with. However, our algorithm  works with a more abstract and operational representation of programs called \emph{probabilistic control-flow graphs (pCFGs)}. The use of pCFGs is standard in probabilistic program analysis~\cite{ChatterjeeFNH18,AgrawalC018}, hence we use them as the primary syntactical representation of programs.

\smallskip\noindent{\bf\em Probabilistic control-flow graphs.} 
A {\em probabilistic control-flow graph (pCFG)} is an ordered tuple $\pCFG=(\locs,V,\Vout,\locinit,\vecinit,\transitions,\guards,\updates)$, where:
\begin{compactitem}
	\item $\locs$ is a finite set of {\em locations};
	\item $V=\{x_1,\dots,x_{|V|}\}$ is a finite set of {\em program variables};
	\item $\Vout = \{x_1,\dots,x_{|\Vout|}\} \subseteq V$ is a finite set of {\em output variables};
	\item $\locinit \in \locs$ is the {\em initial program location} and $\vecinit \in \mathbb{R}^{|V|}$ is the initial variable valuation;
	\item $\transitions\,\subseteq \locs \times \mathcal{P}(\locs)$ is a finite set of {\em transitions}. For each transition $\tau=(\loc,\prob)$, we say that $\loc$ is its {\em source location} and that $\prob:\locs \rightarrow [0,1]$ is a probability distribution over {\em successor locations}.
	\item $\guards$ is a map assigning to each transition $\tau=(\loc,\prob)\in\,\transitions$ a {\em guard} $\guards(\tau)$, which is a predicate over $V$ specifying whether $\tau$ can be executed.
	\item $\updates$ is a map assigning to each transition $\tau=(\loc,\prob)\in\,\transitions$ an {\em update} $\updates(\tau)=(j,u)$ where $j\in\{1,\dots,|V|\}$ is a {\em target variable index} and $u$ is an {\em update element} which can be:
	\begin{compactitem}
		\item the bottom element $u=\bot$, denoting no update;
		\item a Borel-measurable arithmetic expression $u:\mathbb{R}^{|V|}\rightarrow\mathbb{R}$, denoting deterministic update;
		\item a probability distribution $u = \delta$, denoting that variable value is sampled according to $\delta$.
	\end{compactitem}
\end{compactitem}
We assume the existence of a special {\em terminal location} $\locterm$. 
Terminal location $\locterm$ only has one outgoing self-loop transition $\tau=(\locterm,\prob)$ with $\prob(\locterm)=1$, $\guards(\tau) \equiv \text{true}$ and no variable update. 

We require that each location $\loc$ has at least one outgoing transition and that the disjunction of guards of all transitions outgoing from $\loc$ is equivalent to $\mathit{true}$, i.e.~$\bigvee_{\tau=(l,\_)}\guards(\tau)\equiv\mathit{true}$. These assumptions ensure that it is always possible to execute at least one transition and are imposed without loss of generality as we may always introduce an additional transition from $\loc$ to $\locterm$. We also require that guards of two distinct transitions $\tau_1$ and $\tau_2$ outgoing from $\loc$ are {\em mutually exclusive}, i.e.~$\guards(\tau_1)\land\guards(\tau_2)\equiv\mathit{false}$, to ensure that there is no non-determinism in the programming language.

\subsection{Program Semantics}\label{sec:prelimssemantics}

We use operational semantics that views each pCFG as a (general state space) Markov process. This approach is standard in probabilistic program analysis~\cite{ChatterjeeFNH18,KaminskiKMO18}. 
Towards the end of the subsection, we define the \emph{output distribution} of a probabilistic program.

\smallskip\noindent{\bf\em States, paths and runs.} A {\em state} in a pCFG $\pCFG$ is a tuple $(\loc,\mathbf{x})$, where $\loc$ is a location in $\pCFG$ and $\mathbf{x}\in\mathbb{R}^{|V|}$ is a variable valuation. A transition $\tau=(\loc,\prob)$ is {\em enabled} at a state $(\loc',\mathbf{x})$ if \( \loc = \loc' \) and $\mathbf{x}\models\guards(\tau)$. A state $(\loc',\mathbf{x}')$ is a {\em successor} of $(\loc,\mathbf{x})$, if there exists an enabled transition $\tau=(\loc,\prob)$ in $\pCFG$ such that $\prob(\loc') > 0$ and we can obtain $\mathbf{x}'$ by applying the update of $\tau$ to $\mathbf{x}$. The state $(\locinit,\vecinit)$ is the {\em initial state}. A state $(\loc,\mathbf{x})$ is said to be {\em terminal}, if $\loc=\locterm$. We use $\State^\pCFG$ to denote the set of all states~in~$\pCFG$.

A {\em finite path} in $\pCFG$ is a sequence $(\loc_0,\mathbf{x}_0),(\loc_1,\mathbf{x}_1),\dots,(\loc_k,\mathbf{x}_k)$ of states with $(\loc_0,\mathbf{x}_0)=(\locinit,\vecinit)$ and with $(\loc_{i+1},\mathbf{x}_{i+1})$ being a successor of $(\loc_i,\mathbf{x}_i)$ for each $0\leq i\leq k-1$. A state $(\loc,\mathbf{x})$ is {\em reachable} in $\pCFG$ if there exists a finite path in $\pCFG$ whose last state is $(\loc,\mathbf{x})$. A {\em run} (or an {\em execution}) in $\pCFG$ is an infinite sequence of states whose each finite prefix is a finite path. We use $\Fpath^\pCFG$ and $\Run^\pCFG$ to denote the set of all finite paths and all runs in $\pCFG$, respectively. We also use $\Reach^\pCFG$ to denote the set of all reachable states in $\pCFG$.

\smallskip\noindent{\bf\em Next valuation function.} 
Let \( \tau \in \,\transitions \) be a transition and \( \val \) a valuation. By \( \nextv(\tau,\val) \) we denote a random vector representing the successor valuation after \( \tau \) is taken in a state whose current valuation is \( \val \). Formally, let \(  (i,u) = \updates(\tau) \). Then \( \nextv(\tau,\val)[j] = \val[j]  \) for all variable indices \( 1 \leq j \leq |V| \) with \( j \neq i \) that are not updated by the transition, and
\[
\nextv(\tau,\val)[i] = \begin{cases}
\val[i] & \text{if } u = \bot ,\\
u(\val) & \text{if } u \text{ is a Borel-measurable arithmetic expression}, \\
X_\delta & \begin{aligned}[t]
&\text{if } u = \delta \text{ is a probability distribution}\\ &\text{(here \( X_\delta \) is a random variable following the distribution \( \delta \))}.
\end{aligned} 
\end{cases}
\]

\smallskip\noindent{\bf\em Semantics of pCFGs.} A pCFG $\pCFG$ defines a discrete-time Markov process taking values in the set of states of $\pCFG$, whose trajectories correspond to runs in $\pCFG$. Intuitively, the process starts in the initial state $ (\locinit,\vecinit) $ and in each time step it samples the next state along the run from the probability distribution defined by the current state. Suppose that, at time step $ i $, the process is in the state $(\loc_i,\mathbf{x}_i)$. The next state $(\loc_{i+1},\mathbf{x}_{i+1})$ is chosen as follows:
\begin{compactitem}
	\item Let $\tau = (\loc_i,\prob_i)$ be the unique transition enabled at $(\loc_i,\mathbf{x}_i)$. Recall, our assumptions on pCFGs ensure that at each state in $\pCFG$ there is a unique enabled transition.
	\item Sample the successor location $\loc_{i+1}$ from the probability distribution $\prob_i$.
	\item Sample a value of the random vector \( \nextv(\tau, \veca{x}_i) \) to get \( \veca{x}_{i+1} \).
\end{compactitem}

The above intuition can be formalized by a construction of a probability space whose sample space is \( \Run^\pCFG \). The construction is standard (see, e.g.,~\cite{MeynTweedie}) and we omit it. We denote by \( \probm^\pCFG \) the probability measure over the runs of \( \pCFG  \) which results from this construction and which thus formally captures the dynamics intuitively explained above.

\paragraph{Termination.} Our equivalence analysis is restricted to probabilistic programs that terminate almost-surely. This is both a conceptual assumption since we want our probabilistic programs to define valid probability distributions over their outputs, and also a technical assumption required by our approach.
Given a pCFG $\pCFG$, a run $\rho = (\loc_0,\mathbf{x}_0),(\loc_1,\mathbf{x}_1), \dots \in \Run^\pCFG$ is {\em terminating} if it reaches some terminal state. We use $\Termset\subseteq\Run^{\pCFG}$ to denote the set of all terminating runs in $\Run^{\pCFG}$. 
A pCFG $\pCFG$ terminates {\em almost-surely (a.s.)} if $\mathbb{P}^\pCFG[\Termset] = 1$. Automated almost-sure termination proving for linear and polynomial arithmetic probabilistic programs can be achieved by synthesizing a ranking supermartingale (RSM)~\cite{ChakarovS13,ChatterjeeFNH18,ChatterjeeFG16}. We define the {\em termination time} of $\rho$ via $\TimeTerm(\rho) = \inf_{i \geq 0}\{i \mid \loc_i = \locterm \}$, with $\TimeTerm(\rho) = \infty$ if $\rho$ is not terminating. 

\paragraph{Output distribution.} Every a.s.~terminating pCFG defines a probability distribution over its outputs. For every variable valuation $\mathbf{x} \in \mathbb{R}^{|V|}$, let $\mathbf{x}^{\out}$ be the projection of $\mathbf{x}$ to the components corresponding to variables in $\Vout$. Then, for a terminating run $\rho$ that reaches a terminal state $(\locterm,\mathbf{x})$, we say that $\mathbf{x}^{\out}$ is its {\em output variable valuation} (or, simply, its {\em output}). 
An a.s.~terminating pCFG $\pCFG$ defines a probability distribution over the space of all output variable valuations $\mathbb{R}^{|\Vout|}$ as follows. For each Borel-measurable subset $B \subseteq \mathbb{R}^{|\Vout|}$, we define
\[ \Output(B) = \Big\{ \rho \in \Run^\pCFG \mid \rho \textit{ reaches a terminal state } (\locterm,\mathbf{x}) \text{ with } \mathbf{x}^{\out} \in B \Big\}. \]
A {\em output distribution $\mu^{\pCFG}$} of \( \pCFG \) is defined by putting
\[ \mu^{\pCFG}[B] = \mathbb{P}^{\pCFG}\Big[\Output(B) \Big] \]
for each Borel-measurable subset $B$ of \( \Rset^{|\Vout|} \). 
Since $\pCFG$ is a.s.~terminating, we have $\mu^{\pCFG}[\mathbb{R}^{|\Vout|}] = 1$.

\subsection{Kantorovich distance of probability distributions}\label{sec:distances}

To measure the similarity of output distributions, we use the established \emph{Kantorovich distance} (also known as $1$-Wasserstein distance). The definition of this distance is parameterized by a choice of a \emph{metric} in \( \Rset^{|\Vout|} \); this is an Euclidean space and thus can be equipped with a number of standard metrics such as as discrete, $L^1$ (i.e.~Manhattan), $L^2$ (i.e.~Euclidean) or $L^{\infty}$ (i.e.~uniform).

The standard, "primal", definition of Kantorovich distance~\cite{villani2021topics} between two distributions \( \mu_1, \mu_2 \), which involves \emph{couplings} between the two distributions, is somewhat technical and we omit it. However, since we consider distances of \( \Rset^{|\Vout|} \)-valued distributions, we can use an equivalent \emph{dual} definition, which we present below. 






\paragraph{Kantorovich distance: definition.} The Kantorovich distance is only well-defined for pairs of distributions that have \emph{finite first moments} w.r.t.~the underlying metric. Given a metric \( d \colon {\Rset^{|\Vout|}}^2 \rightarrow \Rsetnn \), we say that a probability measure \( \probmalt \) has a \emph{finite first moment} w.r.t. \( d \) if there exists \( \veca{x}_0 \in \Rset^{|\Vout|} \) s.t. the function \( g_{\veca{x}_0} \colon \Rset^{|\Vout|} \rightarrow \Rsetnn \) defined by \( g_{\veca{x}_0}(\veca{x}) = d(\veca{x}_0,\veca{x}) \) satisfies \( \E_\probmalt[g_{\veca{x}_0}] < \infty \). Note that due to the triangle inequality property of metrics, \( \E_\probmalt[g_{\veca{x}_0}] < \infty \) iff \( \E_\probmalt[g_{\veca{y}}] < \infty \) for all \( \veca{y} \in \Rset^{|\Vout|} \).

\begin{definition}[Kantorovich distance of output distributions]
Let \( \pCFG_1 \), \( \pCFG_2 \) be two pCFGs with the same set of output variables \( \Vout \).
Further, let \( d \) be a metric in \( \Rset^{|\Vout|} \) such that \( \mu^{\pCFG_1} \) and \( \mu^{\pCFG_2} \) have finite first moments w.r.t. \( d \). The {\em Kantorovich (or $1$-Wasserstein) distance} between \( \mu^{\pCFG_1} \) and \( \mu^{\pCFG_2} \) is defined~via
	\[ \K_d(\mu^{\pCFG_1} ,  \mu^{\pCFG_2}) = \sup_{f \in L^1_d(\Rset^{|\Vout|})} \Big|\mathbb{E}_{\mu_1}[f] - \mathbb{E}_{\mu_2}[f] \Big|, \]
	where 
	\[ L^1_d(\Rset^{|\Vout|}) = \Big\{f:\Rset^{|\Vout|}\rightarrow \mathbb{R} \,\Big|\, |f(x) - f(y)| \leq d(x,y) \text{ for all } x,y\in\Omega\Big\} \]
	is the set of all $1$-Lipschitz continuous functions defined over the metric space $(\Rset^{|\Vout|},d)$, and $\mathbb{E}_{\mu_1}$ and $\mathbb{E}_{\mu_2}$ denote expectation operators with respect to $\mu_1$ and $\mu_2$.
\end{definition}

When defined with respect to the discrete metric (which assigns 0 distance to pairs of identical elements and unit distance to all distinct pairs), Kantorovich distance is equal to another well known distance of probability distributions: the Total Variation distance~\cite{villani2021topics}, which has been previously used in verification of finite-state probabilistic systems~\cite{DBLP:conf/csl/ChenK14,DBLP:conf/icalp/Kiefer18}. 
Moreover, for finite-state probabilistic models, the notion of simulation distance is based on the notion of \emph{optimal transport}~\cite{LARSEN19911,TracolDZ11}, and Kantorovich distance generalizes this notion to infinite-state models. The survey paper~\cite{deng2009kantorovich} gives an overview of various uses of Kantorovich distance in probabilistic verification.



\section{Problem Statement}\label{sec:problem}


In what follows, let $\pCFG_1$
and $\pCFG_2$ 
 be two pCFGs. Since we can only compare two probability distributions if they are defined over the same space, we require that the two pCFGs share a common output variable set $\Vout$. Furthermore, we assume that both $\pCFG_1$ and $\pCFG_2$ are a.s.~terminating.

\subsection{Equivalence Refutation Problem}



We say that $\pCFG_1$ and $\pCFG_2$ are {\em output equivalent}, if for every Borel-measurable set $B \subseteq \mathbb{R}^{|\Vout|}$ we have 
\begin{equation}\label{eq:equivalence}
	\mu^{\pCFG_1}[B] = \mu^{\pCFG_2}[B].
\end{equation}
(Recall that $\mu^{\pCFG_1}$ and $\mu^{\pCFG_2}$ denote output distributions of $\pCFG_1$ and $\pCFG_2$, respectively.)


\medskip
\begin{adjustwidth}{0.5cm}{}
	{\bf Problem~1 (Equivalence refutation problem).} Given two a.s.~terminating pCFGs $\pCFG_1$ and $\pCFG_2$ with the same output variable set $\Vout$, prove that $\pCFG_1$ and $\pCFG_2$ are \underline{not} output equivalent.
\end{adjustwidth}

\subsection{Similarity Refutation Problem}


The similarity refutation problem is parameterized by a metric over the output space which gives rise to a Kantorovich distance of distributions over this space.
%
Our theoretical results in this work are applicable to any metric. Our algorithmic approach will consider standard metrics such as discrete, $L^1$ (i.e.~Manhattan), $L^2$ (i.e.~Euclidean) or $L^{\infty}$ (i.e.~uniform). We provide the definition of each of these metrics in Section~\ref{app:similarityalgo} in the Supplementary material.

Let \( d \) be a metric over \( \Rset^{|\Vout|} \) such that the output distributions $\mu^{\pCFG_1}, \mu^{\pCFG_2}$ have finite first moments w.r.t. \( d \). We say that $\pCFG_1$ and $\pCFG_2$ are {\em $\epsilon$-output close}, if
	\begin{equation}\label{eq:distance}
		\K_d\Big(\mu^{\pCFG_1},\mu^{\pCFG_2}\Big) < \epsilon.
	\end{equation}
The definition of the similarity refutation problem follows straightforwardly.



\medskip
\begin{adjustwidth}{0.5cm}{}
	{\bf Problem~2 (Similarity refutation problem).} Given two a.s.~terminating pCFGs $\pCFG_1$ and $\pCFG_2$ with the same output variable set $\Vout$, a metric $d$ over $\mathbb{R}^{|\Vout|}$ such that \( \mu^{\pCFG_1} \) and \( \mu^{\pCFG_2} \) have finite first moments w.r.t \( d \), and $\epsilon>0$, prove that $\pCFG_1$ and $\pCFG_2$ are \underline{not} $\epsilon$-output close.
\end{adjustwidth}
\medskip


\paragraph{Finite first moment assumption.} The Similarity refutation problem assumes that the output distributions of the two programs have finite first moments w.r.t. the output state metric. Our algorithm (see Section~\ref{sec:algo}) checks this assumption (or more precisely, its sufficient conditions) automatically. Section~\ref{app:moments} in the supplementary material contains a discussion of what to do when we aim to analyze a pair of programs that violate the assumption.

\section{Martingale-Based Refutation Rules}\label{sec:expectationsupermartingales}

Our approach to equivalence and similarity refutation is based on the notions of upper and lower expectation supermartingales, which generalize and unify cost supermartingales~\cite{Wang0GCQS19} and super/sub-invariants~\cite{HarkKGK20}. Given an a.s.~terminating probabilistic program and a function \( f \) over its output variables, upper expectation supermartingales provide a sound proof rule for deriving upper bounds on the expected value of \( f \) at output in probabilistic programs, and similarly for lower expectation submartingales and lower bounds. 

In this section, we start by fixing the necessary terminology. Then, we state proof rules which subsume and unify the proof rules of~\cite{Wang0GCQS19,HarkKGK20}. Finally, we state our proof rules for equivalence and similarity refutation in probabilistic program pairs.

\subsection{Expectation Supermartingales and Submartingales: Definition} 

\noindent{\bf\em State and predicate functions, invariants.} Let $\pCFG=(\locs,V,\Vout,\locinit,\vecinit,\transitions,\guards,\updates)$ be a pCFG:
\begin{compactitem}
	\item A {\em state function} $\lem$ in $\pCFG$ is a function which to each location $\loc \in \locs$ assigns a Borel-measurable function $\lem(\loc):\mathbb{R}^{|V|} \rightarrow \mathbb{R}$ over program variables. We interchangeably use $\lem(\loc)(\mathbf{x})$~and~$\lem(\loc,\mathbf{x})$.
	\item A {\em predicate function} $\Pi$ in $\pCFG$ is a function which to each location $\loc \in \locs$ assigns a predicate $\Pi(\loc)$ over program variables. It naturally induces a set of states $\{(\loc,\mathbf{x}) \mid \mathbf{x} \models \Pi(\loc)\}$. With a slight abuse of notation, we also use $\Pi$ to refer to this set of states.
	\item A predicate function \( \Pi \) is an {\em invariant} if \( \Pi \) contains all reachable states in $\pCFG$, i.e.~if for each reachable state $(\loc,\mathbf{x}) \in \Reach^{\pCFG}$ we have $\mathbf{x} \models I(\loc)$.
\end{compactitem}

\smallskip\noindent{\bf\em Upper expectation supermartingales.} 
We now define upper expectation supermartingales (UESMs). Consider a pCFG $\pCFG$ and let $f: \mathbb{R}^{|\Vout|} \rightarrow \mathbb{R}$ be a Borel-measurable function over its outputs. 
A UESM for $f$ is a state function $U_f$ that satisfies certain conditions in every reachable state.

Since it is generally not feasible to compute the set of all reachable states in a program, we define UESMs with respect to a supporting invariant that over-approximates the set of all reachable states. This is done with later automation in mind, and our algorithmic approach in Section~\ref{sec:algo} will first automatically synthesize this supporting invariant (or, alternatively, invariants can be provided by the user) before proceeding to the synthesis of an UESM. 
Example~\ref{ex:ulesm} shows an example UESM.

\begin{definition}[Upper expectation supermartingale (UESM)]\label{def:uesm}
	Let $\pCFG=(\locs,V,\Vout,\locinit,\vecinit,\transitions,\guards,\updates)$ be an a.s.~terminating pCFG, $I$ be an invariant in $\pCFG$ and $f:\mathbb{R}^{|\Vout|} \rightarrow \mathbb{R}$ be a Borel-measurable function over the output variables of $\pCFG$. An {\em upper expectation supermartingale (UESM) for $f$} with respect to the invariant $I$ is a state function $U_f$ satisfying the following two conditions:
	\begin{compactenum}
		\item {\em Zero on output.} For every $\mathbf{x} \models I(\locterm)$, we have $U_f(\locterm,\mathbf{x}) = 0$.
		\item {\em Expected $f$-decrease.} For every location $\loc\in\locs$ , transition $\tau = (\loc,\prob) \in\,\transitions$, and valuation \( \val \) s.t. $\val \models I(\loc) \land \guards(\tau)$, we require the following: for \( \veca{N} = \nextv(\tau,\val)\) it holds
		 \begin{equation}
		 \label{eq:ufexp}
		 U_f(\loc,\mathbf{x}) \geq \sum_{\loc'\in\locs}\prob(\loc') \cdot \E[U_f(\ell', \veca{N}) + f(\veca{N}^\out) ] - f(\mathbf{x}^{\out})
		 \end{equation}
		(where \( \veca{N}^\out \) is the projection of the random vector \( \veca{N} \) onto the \( V_\out \)-indexed components). Intuitively, this condition requires that, in any step of computation, any increase in the \( f \)-value of the current valuation (projected onto the output variables) is matched, in expectation, by the decrease of the \( U_f \)-value.
	\end{compactenum}
\end{definition}

\noindent{\bf\em Lower expectation submartingales.} A lower expectation submartingale is defined analogously as an UESM, with the expected \( f \)-decrease condition replaced by the dual expected \emph{\( f \)-increase} condition. Example~\ref{ex:ulesm} shows an example LESM.

\begin{definition}[Lower expectation submartingale (LESM)]\label{def:lesm}
	Let $\pCFG=(\locs,V,\Vout,\locinit,\vecinit,\transitions,\guards,\updates)$ be an a.s.~terminating pCFG, $I$ be an invariant in $\pCFG$ and $f:\mathbb{R}^{|\Vout|} \rightarrow \mathbb{R}$ be a Borel-measurable function over the output variables of $\pCFG$. A {\em lower expectation submartingale (LESM) for $f$} with respect to the invariant $I$ is a state function $L_f$ satisfying the following two conditions:
	\begin{compactenum}
		\item {\em Zero on output.} For every $\mathbf{x} \models I(\locterm)$, we have $L_f(\locterm,\mathbf{x}) = 0$.
		\item {\em Expected $f$-increase.} For every location $\loc\in\locs$ , transition $\tau = (\loc,\prob) \in\,\transitions$, and valuation \( \val \) s.t. $\val \models I(\loc) \land \guards(\tau)$, we require the following: for \( \veca{N} = \nextv(\tau,\val)\) it holds
		 \begin{equation}
		 \label{eq:lfexp}
		 L_f(\loc,\mathbf{x}) \leq \sum_{\loc'\in\locs}\prob(\loc') \cdot \E[L_f(\ell', \veca{N}) + f(\veca{N}^\out) ] - f(\mathbf{x}^{\out}).
		 \end{equation}
	\end{compactenum}
\end{definition}

\subsection{Expectation Bounds via U/LESMs}\label{sec:proofrulesbounds}

In this subsection, we state Theorem~\ref{thm:ulesm-exp-soundness}, which shows that under certain conditions, a U/LESM for a function \( f \) provides an upper/lower bound on the expected value of \( f \) at output. The theorem is used to prove soundness of our proof rules for equivalence and similarity refutation in Section~\ref{sec:proofequivalencesimilarity}. 

The result of Theorem~\ref{thm:ulesm-exp-soundness} subsumes and unifies the proof rules of~\cite{Wang0GCQS19,HarkKGK20}. Both these papers use Optional Stopping Theorem (OST)~\cite{Williams91} to formulate conditions under which U/LESMs provide sound proof rules for computing expectation bounds. The proof rule of~\cite{HarkKGK20} uses the classical OST~\cite{Williams91} (though only for lower bounds, see the discussion below), whereas the work of~\cite{Wang0GCQS19} derives what they call Extended OST to relax and replace some of the classical OST conditions.

Our approach to the formulation and proof of Theorem~\ref{thm:ulesm-exp-soundness} differs from the proof rules in~\cite{Wang0GCQS19,HarkKGK20} in the following aspects: First, in~\cite{Wang0GCQS19}, the upper/lower cost supermartingales provided bounds on the expected value of a single~\emph{cost} variable whereas we consider arbitrary functions of the output variables. Second, the approach in~\cite{HarkKGK20} considers functions \( f \) taking values in the \emph{non-negative} and \emph{extended} real interval \( [0,\infty] \). In such a case, the space of all such functions forms a complete lattice, which allows the use of Park induction~\cite{park1969fixpoint} to obtain upper bounds. In contrast, we work with functions taking values in \( (-\infty,\infty) \), which precludes such approach. We need to work with the co-domain \( (-\infty,\infty) \) to achieve automation: our algorithm, presented in Section~\ref{sec:algo}, works with functions represented via polynomials, which in general have this co-domain.  Hence, we use OST for both upper and lower bounds.



We start by stating the conditions under which U/LESMs provide the required bounds, which we call the {\em OST-soundness conditions}. The first three conditions are conditions imposed by the classical OST~\cite{Williams91}, whereas the fourth condition is imposed by the Extended OST~\cite{Wang0GCQS19}. In the following, we use the notion of \emph{conditional expectation.} For the sake of brevity, we omit its formal definition. Intuitively, when dealing with a pCFG \( \pCFG \) and a random variable \( X \) over the runs of \( \pCFG \), we denote by \( E[X\mid \mathcal{F}_t] \) the \emph{conditional expected value} of \( X \) given the knowledge of the first \( t \) steps of \( \pCFG \)'s run.

\begin{definition}[OST-soundness]\label{def:ostsound}
Let \( \pCFG \) be a pCFG, \( \eta \) be a state function in \( \pCFG \), and \( f \colon \Rset^{|V_\out|} \rightarrow \Rset\) be a Borel measurable function. Denote by \( Z_i(\rho) \) the \( i \)-the state along a run \( \rho \), and let \( Y_i\) be defined by \( Y_i := \eta(Z_i) + f(\veca{X}_i^\out)\) for any \( i \geq 0 \). We say that the tuple \( (\pCFG, \eta, f) \) is \emph{OST-sound} if \( \E[|Y_i|] < \infty \) for every \( i\geq 0 \) and moreover, at least one of the following conditions (C1)--(C4) holds:
\begin{itemize}
\item[(C1)] There exists a constant \( c \) such that \( \TimeTerm \leq c \) with probability 1 (i.e., the termination time of the program is uniformly bounded).
\item[(C2)] There exists a constant \( c \) such that for each \( t \in \Nset \) and each run \( \rho \) it holds that \[ |Y_{\min\{t, \TimeTerm(\rho)\}}(\rho)| \leq c \] (i.e., \( Y_{i}(\rho) \) is uniformly bounded from below and above up until the point of termination).
\item[(C3)] \( \E[\TimeTerm] < \infty \), \( \E[|Y_0|] < \infty \), and there exists a constant \( c \) such that for every \( t \in \Nset \) it holds \( \E[|Y_{t+1}-Y_{t}|\mid \sigmaAlg_t] \leq c \) (i.e., the expected one-step change of \( Y_i \) is uniformly bounded over the program runtime, even if conditioned by the whole past history of the program).
\item [(C4)] There exist real numbers \( M, c_1, c_2, d \) such that (i) for all sufficiently large \( n \in \Nset \) it holds \( \probm(\TimeTerm > n) \leq c_1 \cdot e^{-c_2 \cdot n} \); and (ii) for all \( t \in \Nset \) it holds \( |Y_{n+1} - Y_n| \leq M\cdot n^d \).
\end{itemize}
\end{definition}

Our algorithm presented in Section~\ref{sec:algo} will automatically enforce OST-soundness. We are now ready to state the U/LESM soundness theorem. 

\begin{restatable}[Soundness of U/LESMs]{theorem}{soundness}\label{thm:ulesm-exp-soundness}
Let \( \pCFG=(\locs,V,\Vout,\locinit,\vecinit,\transitions,\guards,\updates) \) be an a.s. terminating pCFG with output distribution \( \mu^\pCFG \) and \( f \colon \Rset^{|V_\out|} \rightarrow \Rset\) a Borel measurable function over the outputs of $\pCFG$. Let \( U_f \) and \( L_f \) be an upper (respectively lower) expectation supermartingale for \( f \) w.r.t. some invariant. Assume that \( (\pCFG, U_f, f) \) and \( (\pCFG, L_f, f) \) are OST-sound. Then \( \E_{\mu^\pCFG}[f(\val^\out)] \) is well-defined and
\begin{align*}
U_f(\locinit, \vecinit) + f(\vecinit^\out) &\geq \E_{\mu^\pCFG}[f(\val^\out)], \\
L_f(\locinit, \vecinit) + f(\vecinit^\out) &\leq \E_{\mu^\pCFG}[f(\val^\out)]. 
\end{align*}

\end{restatable}
\begin{proof}[Proof (sketch).]
Let \( Z_n \) denote the \( n \)-th state along a run of \( \pCFG \) and \( \veca{X}_n \) denotes the \( n \)-th valuation encountered along a run. For \( L_f \), we define a stochastic process \( Y = (Y_n)_{n=0}^{\infty} \) by putting \( Y_n := L_f(Z_n) + f(\veca{X}_n^\out)\). The inequality for \( L_f \) follows from application of the (extended) optional stopping theorem~\cite{Williams91, Wang0GCQS19} to \( Y \), which is permissible due to the OST-soundness assumption. The argument for \( U_f \) is analogous. Full proof can be found in Section~\ref{sec:soundproof} of the Supplementary material.
\end{proof}

\subsection{Proof Rules for Equivalence and Similarity Refutation}\label{sec:proofequivalencesimilarity}

We now show how to use the results in the previous section to derive refutation rules for the equivalence and similarity problems. 
 Example~\ref{ex:uelsmequiv} shows an application of this proof rule for equivalence refutation, and Example~\ref{ex:uelsmsimil} for similarity refutation.

\begin{theorem}[Soundness of equivalence and similarity refutation]
\label{thm:ulesm-refute}
Consider two a.s.~terminating pCFGs $\pCFG_1=(\locs^1,V^1,\Vout,\locinit^1,\vecinit^1,\transitions^1,\guards^1,\updates^1)$ and $\pCFG_2=(\locs^2,V^2,\Vout,\locinit^2,\vecinit^2,\transitions^2,\guards^2,\updates^2)$.
Assume that there exists a Borel-measurable function \( f \colon \Rset^{|V_\out|} \rightarrow \Rset \) and two state functions, \( U_f \) for \( \pCFG_1 \) and \( L_f \) for \( \pCFG_2 \), such that the following holds:
\begin{compactitem}
\item \( U_f \) is a UESM for \( f \) in \( \pCFG_1 \) such that \( (\pCFG_1, U_f, f) \) is OST-sound;
\item \( L_f \) is an LESM for \( f \) in \( \pCFG_2 \) such that \( (\pCFG_2, L_f, f) \) is OST-sound;
\item \( U_f(\locinit^1, \vecinit^1) + f((\vecinit^1)^\out) < L_f(\locinit^2, \vecinit^2) + f((\vecinit^2)^\out) \).
\end{compactitem}
Then \( \pCFG_1 \) and \( \pCFG_2 \) do not define equivalent output distributions.

Moreover, if \( f \) is \( 1 \)-Lipschitz continuous under a metric \( d \) of the output space \( \Rset^{|V_\out|} \), then 
\[\K_d\Big(\mu^{\pCFG_1},\mu^{\pCFG_2}\Big) \geq L_f(\locinit^2, \vecinit^2) + f((\vecinit^2)^\out) - U_f(\locinit^1, \vecinit^1) - f((\vecinit^1)^\out).\]
\end{theorem}

\begin{proof}
	
	From Theorem~\ref{thm:ulesm-exp-soundness} we have
	\begin{align}
		\E_{\mu^{\pCFG_1}}[f(\val^\out)] &\leq U_f(\locinit^1, \vecinit^1) + f((\vecinit^1)^\out)\nonumber \\
		&< L_f(\locinit^2, \vecinit^2) + f((\vecinit^2)^\out) \leq \E_{ \mu^{\pCFG_2}}[f(\val^\out)] \label{eq:refute}
	\end{align}
	Hence, $\E_{\mu^{\pCFG_1}}[f(\val^\out)] < \E_{\mu^{\pCFG_2}}[f(\val^\out)]$, and so the output distributions $\mu^{\pCFG_1}$ and $\mu^{\pCFG_2}$ are not equivalent (otherwise, any measureable \( f \) would have the same expectation under both measures).
	
	The second part follows directly from~\eqref{eq:refute} and from the definition of the Kantorovich distance.
\end{proof}

To conclude this section, we highlight that searching for a Borel-measurable function \( f \colon \Rset^{|V_\out|} \rightarrow \Rset \) such that $\E_{\mu^{\pCFG_1}}[f(\val^\out)] \neq \E_{\mu^{\pCFG_2}}[f(\val^\out)]$ yields {\em both sound and complete} proof rule for refuting equivalence of output distributions. While soundess follows from Theorem~\ref{thm:ulesm-refute}, to prove completeness suppose that two output distributions are not equivalent. Then, there exists an event $A$ over outputs such that $\mu^{\pCFG_1}[A]\neq \mu^{\pCFG_1}[A]$. Hence, with $f$ being the indicator function $I(A)$ of the event $A$, which is indeed a Borel-measurable function, we have $E_{\mu^{\pCFG_1}}[I(A)]\neq E_{\mu^{\pCFG_1}}[I(A)]$.

\section{Automated Constraint Solving-based Algorithm}\label{sec:algo}

In this section, we present our algorithms for automated equivalence and similarity refutation.
In the sequel, let $\pCFG_1=(\locs^1,V^1,\Vout,\locinit^1,\vecinit^1,\transitions^1,\guards^1,\updates^1)$ and $\pCFG_2=(\locs^2,V^2,\Vout,\locinit^2,\vecinit^2,\transitions^2,\guards^2,\updates^2)$ be two a.s.~terminating pCFGs with a common output variable set $\Vout$.

\smallskip\noindent{\bf\em Assumptions.} Our algorithms impose the following assumptions:
\begin{compactitem}
	\item {\em Polynomial programs.} We consider probabilistic programs in which all arithmetic expressions are {\em polynomials} over program variables. Furthermore, by introducing dummy variables for expressions appearing in transition guards, we without loss of generality assume that arithmetic expressions appearing in transition guards are linear.
	\item {\em Finite moments of probability distributions.} We assume that each probability distribution \( \delta \) appearing in sampling instructions  has {\em  finite moments} which are accessible to the algorithm, i.e. for each \( p \in \Nset \), the \( p \)-th moment $m_{\delta}(p) = \mathbb{E}_{X \sim \delta}[|X|^p]$ is finite and can be computed by the algorithm. This is a standard assumption in static probabilistic program analysis and allows sampling instructions from most standard probability distributions.
	\item {\em Linear invariants.} Recall, 
	in Definition~\ref{def:uesm} and Definition~\ref{def:lesm} we defined U/LESMs with respect to supporting invariants. We assume that we are provided with {\em linear invariants} $I_1$ and $I_2$ for $\pCFG_1$ and $\pCFG_2$, respectively. Linear invariant generation is a well-studied problem in program analysis; in our implementation, we use the methods of~\cite{FeautrierG10,SankaranarayananSM04} to synthesize supporting linear invariants $I_1$ and $I_2$.
	\item {\em OST-soundness.} Recall from Section~\ref{sec:proofrulesbounds} that we need to impose one of the OST-soundness conditions in Definition~\ref{def:ostsound} on each pCFG for the proof rules based on U/LESMs to be sound. These conditions impose restrictions on the pCFG as well as on the function on outputs and the U/LESMs that we need to synthesize. In what follows, we state the restrictions imposed on pCFGs by each of the OST-conditions. In principle, a different restriction can be imposed on each of the two pCFGs. To streamline the presentation, we consider imposing the same condition on both pCFGs, in an order of preference specified in the next subsection. Then, depending on the OST-condition that the algorithm uses for the pCFGs, we will also constrain our output functions and U/LESMs to satisfy the corresponding restrictions. We use the same enumeration of OST-conditions as in Definition~\ref{def:ostsound}. We use $\TimeTerm_1$ and $\TimeTerm_2$ to denote random variables defined by termination time in $\pCFG_1$ and $\pCFG_2$:
	\begin{itemize}
		\item[(C1)] The pCFGs have bounded termination time, i.e.~there exists $c>0$ s.t.~$\TimeTerm_1(\rho) \leq c$ for all runs $\rho$ in $\pCFG_1$ and $\TimeTerm_2(\rho) \leq c$ for all runs $\rho$ in $\pCFG_2$. To enforce this condition, it suffices to restrict our attention to programs in which all loops are statically bounded. 
		\item[(C2)] This condition imposes no restrictions on pCFGs.
		\item[(C3)] The pCFGs have bounded {\em expected} termination time, i.e.~$\mathbb{E}^{\pCFG_i}[\TimeTerm_i] < \infty$ for $i\in\{1,2\}$. To verify this, it suffices to synthesize a ranking supermartingale (RSM)~\cite{ChakarovS13}. Automated synthesis of RSMs in polynomial programs was considered in~\cite{ChatterjeeFNH18,ChatterjeeFG16}.
		\item[(C4)] There exist real numbers \( c_1, c_2 \) such that, for all sufficiently large \( n \in \Nset \), it holds \( \probm(\TimeTerm_i > n) \leq c_1 \cdot e^{-c_2 \cdot n} \) for $i\in\{1,2\}$. It was shown in~\cite{Wang0GCQS19} that, in polynomial programs, to verify the first condition it suffices to synthesize an RSM as in (C3).
	\end{itemize}
\end{compactitem}

\subsection{Algorithm for the Equivalence Refutation Problem}\label{sec:algoequiv}

\noindent{\bf\em Algorithm outline.} Our algorithm uses constraint solving-based synthesis to simultaneously compute a function $f$ over output variables $\Vout$, an UESM $U^1_f$ for $f$ in $\pCFG_1$ and an LESM $L^2_f$ for $f$ in $\pCFG_2$. The algorithm proceeds in four steps. First, it fixes symbolic polynomial templates for $f$, $U^1_f$ and $L^2_f$. Second, it collects the defining constraints of UESMs and LESMs, the equivalence refutation constraint, and the constraints that encode OST-condition restrictions. Third, it translates these constraints into a linear programming (LP) instance. Fourth, it uses an LP solver to solve the resulting LP instance, with each solution giving rise to a valid triple of $f$, $U^1_f$ and $L^2_f$.

Our algorithm builds on classical constraint solving-based methods for static analysis of polynomial (probabilistic) programs for termination~\cite{ChatterjeeFG16}, reachability~\cite{AsadiC0GM21}, safety~\cite{Chatterjee0GG20} or cost~\cite{Wang0GCQS19,ZikelicCBR22} properties. Hence, we keep our exposition brief and focus on the Step~2 of our algorithm which contains the main algorithmic novelty, since it collects the defining constraints of U/LESMs and the relational constraint for equivalence refutation. 

\smallskip\noindent{\bf\em Algorithm parameters.} Our algorithm takes as an input a natural number parameter $d \in \mathbb{N}$ which denotes the maximal degree of polynomials that the algorithm uses for synthesis. Also, it determines which of the OST-soundness conditions to impose on the pCFGs as follows:
\begin{compactenum}
	\item If the pCFGs have bounded termination time (e.g.~they only contain statically bounded loops), then our algorithm imposes (C1) on them since this condition does not introduce any constraints on $f$, $U^1_f$ and $L^2_f$.
	\item Else, if the pCFGs have bounded updates, i.e.~there exists $M>0$ such that every update element in each pCFG changes variable value by at most $M$, then our algorithm imposes (C4) on them. This is because it was shown in~\cite{Wang0GCQS19} that, for a polynomial program with bounded updates, \( (\pCFG, \lem, f) \) is OST-sound with (C4) satisfied for state function $\lem$ and output function $f$. Hence, the algorithm need not introduce any constraints on $f$, $U^1_f$ and $L^2_f$.
	\item Else, if the pCFGs have bounded expected termination time (e.g.~the method of~\cite{ChatterjeeFG16} successfully synthesizes RSMs), then our algorithm imposes (C3) on them as this condition introduces milder constraints on $f$, $U^1_f$ and $L^2_f$ than (C2).
	\item Else, our algorithm imposes (C2) on the pCFGs.
\end{compactenum}

\smallskip\noindent{\bf\em Step~1: Symbolic templates.} The algorithm fixes a symbolic polynomial template for $f$ in the form of a symbolic polynomial of degree at most $d$ over output variables $\Vout$. It also fixes a symbolic polynomial template for the UESM $U^1_f$ for $f$ in $\pCFG_1$, in the form of a symbolic polynomial $U^1_f(\loc)$ of degree at most $d$ over program variables $V^1$ for each location $\loc\in\locs^1$. A template for the LESM $L^2_f$ for $f$ in $\pCFG_2$ is fixed analogously.

This is formally done as follows. Let $\Mono_d(V)$ and $\Mono_d(\Vout)$ denote the sets of all monomials of degree at most $d$ over the variables $V$ and $\Vout$, respectively. The templates for $f$, for $U^1_f(\loc)$ at each location $\loc\in\locs^1$ and for $L^2_f(\loc)$ at each location $\loc\in\locs^2$ are respectively defined by fixing the following symbolic polynomial expressions
\[ \sum_{m \in \Mono_d(\Vout)} f_m \cdot m,\quad\quad \sum_{m \in \Mono_d(V)} u^\loc_m \cdot m,\quad\quad \sum_{m \in \Mono_d(V)} l^\loc_m \cdot m,\]
where $f_m$, $u^\loc_m$ and $l^\loc_m$ are a real-valued {\em symbolic template variables} for each $m$ and $\loc$.

\smallskip\noindent{\bf\em Step~2: Constraint collection.} The algorithm collects all defining constraints for $U^1_f$ and $L^2_f$ to be a UESM and an LESM, the equivalence refutation constraint and the OST-soundness constraints. For every collected constraint that contains an expectation operator, the algorithm symbolically evaluates the expected values to obtain expectation-free polynomial expressions over symbolic template and program variables. This is possible since all involved expressions are polynomials over program variables, all moments of probability distributions are finite and can be computed, and moreover, the expectations integrate over a single variable in the polynomial expression (since in pCFGs, at most one variable is updated in each step).
\begin{compactenum}
	\item {\em UESM constraints.} By Definition~\ref{def:uesm}, an UESM must satisfy the Zero on output condition and the Expected $f$-decrease condition. Both constraints are defined with respect to the supporting invariant $I_1$. As explained above, such invariants can be synthesized automatically, e.g. by methods of \cite{FeautrierG10,SankaranarayananSM04}. The algorithm collects the following constraints:
	\begin{itemize}
		\item {\em Zero on output.} $\forall \mathbf{x}\in\mathbb{R}^{|V^1|}.\, \mathbf{x} \models I_1(\locterm^1) \Rightarrow U^1_f(\locterm^1,\mathbf{x}) = 0$.
		\item {\em Expected $f$-decrease.} For every transition $\tau = (\loc,\prob) \in \,\transitions^1$ and for \( \veca{N} = \nextv(\tau,\val)\):
		\[ \forall \mathbf{x}\in\mathbb{R}^{|V^1|}.\, \mathbf{x} \models I_1(\loc) \land \guards^1(\tau) \Rightarrow U^1_f(\loc,\mathbf{x}) \geq \sum_{\loc'\in\locs^1}\prob(\loc') \cdot \E[U^1_f(\ell', \veca{N}) + f(\veca{N}^\out) ] - f(\mathbf{x}^{\out}) \]
	\end{itemize}
	\item{\em LESM constraints.} Analogously, the algorithm collects constraints for $L^2_f$ to be an LESM for $f$ in $\pCFG_2$. As in Definition~\ref{def:lesm}, constraints are defined w.r.t.~the supporting invariant $I_2$.
	\begin{itemize}
		\item {\em Zero on output.} $\forall \mathbf{x}\in\mathbb{R}^{|V^2|}.\, \mathbf{x} \models I_2(\locterm^2) \Rightarrow L^2_f(\locterm^2,\mathbf{x}) = 0$.
		\item {\em Expected $f$-increase.} For every transition $\tau = (\loc,\prob) \in \,\transitions^2$ and for \( \veca{N} = \nextv(\tau,\val)\):
		\[ \forall \mathbf{x}\in\mathbb{R}^{|V^2|}.\, \mathbf{x} \models I_2(\loc) \land \guards^2(\tau) \Rightarrow L^2_f(\loc,\mathbf{x}) \leq \sum_{\loc'\in\locs^2}\prob(\loc') \cdot \E[L^2_f(\ell', \veca{N}) + f(\veca{N}^\out) ] - f(\mathbf{x}^{\out}) \]
	\end{itemize}
	\item {\em Equivalence refutation constraint.} The algorithm collects the equivalence refutation constraint, according to Theorem~\ref{thm:ulesm-refute}:
	\[ U_f(\locinit^1, \vecinit^1) + f((\vecinit^1)^\out) < L_f(\locinit^2, \vecinit^2) + f((\vecinit^2)^\out) \]
	\item {\em OST-soundness constraints.} Finally, the algorithm collects the constraints for OST-soundness conditions in Definition~\ref{def:ostsound}. Depending on which of the conditions $(C1)-(C3)$ in Definition~\ref{def:ostsound} we impose (see Algorithm parameters above), we collect the following constraints:
	\begin{itemize}
		\item[(C1)] No additional constraints are necessary.
		\item[(C4)] No additional constraints are necessary.
		\item[(C2)] We require that there exists a constant $C>0$ such that the absolute value of the sum of the U/LESM and $f$ is bounded from above by $C$ at every reachable state. Thus, we introduce an additional symbolic variable for $C$, collect the constraint $C>0$ and collect the following constraints for each $\loc\in\locs^1$ and $\loc\in\locs^2$, respectively:
		\begin{equation*}
			\begin{split}
				&\forall \mathbf{x}\in\mathbb{R}^{|V^1|}.\, \mathbf{x} \models I_1(\loc) \Rightarrow \Big| U^1_f(\loc,\mathbf{x}) + f(\mathbf{x}^\out) \Big| \leq C \\
				&\forall \mathbf{x}\in\mathbb{R}^{|V^2|}.\, \mathbf{x} \models I_2(\loc) \Rightarrow \Big| L^2_f(\loc,\mathbf{x}) + f(\mathbf{x}^\out) \Big| \leq C
			\end{split}
		\end{equation*}
	\item[(C3)] We require that there exists a constant $C>0$ such that the sum of the U/LESM and $f$ has bounded expected one-step change at every reachable state.  However, this condition yields a constraint which is not of the form as in eq.~\eqref{eq:handelmaneq} that is needed in Step~3 for reduction to an LP instance. In order to allow for an automated synthesis by reduction to LP, we instead collect a stricter condition of bounded {\em maximal} one-step change at every reachable state. In particular, we introduce a symbolic variable for $C$, collect~$C>0$ constraint, and collect the following constraint for each $\tau=(\loc,\prob) \in \,\transitions^1$ and $\loc'\in L^1$ with $\prob(\loc') > 0$, and for each $\tau=(\loc,\prob) \in \,\transitions^2$ and $\loc'\in L^2$ with $\prob(\loc') > 0$, respectively:
	\begin{equation*}
		\begin{split}
			&\forall \mathbf{x}\in\mathbb{R}^{|V^1|},\veca{N}\in\support(\veca{N}).\, \mathbf{x} \models I_1(\loc) \land \guards^1(\tau) \Rightarrow \Big| U^1_f(\loc,\mathbf{x}) + f(\mathbf{x}^{\out}) - U^1_f(\ell', \veca{N}) - f(\veca{N}^\out) \Big| \leq C \\
			&\forall \mathbf{x}\in\mathbb{R}^{|V^2|},\veca{N}\in\support(\veca{N}). \, \mathbf{x} \models I_2(\loc) \land \guards^2(\tau) \Rightarrow \Big| L^2_f(\loc,\mathbf{x}) + f(\mathbf{x}^{\out}) - L^2_f(\ell', \veca{N}) - f(\veca{N}^\out) \Big| \leq C
		\end{split}
	\end{equation*}
	\end{itemize}
\end{compactenum}

\smallskip\noindent{\bf\em Step~3: Conversion to an LP instance.} This step of our algorithm is analogous to~\cite{ChatterjeeFG16,Wang0GCQS19,AsadiC0GM21,ZikelicCBR22}. Observe that the equivalence refutation constraint is a linear and purely existentially quantified constraint over the symbolic template variables of $f$, $U^1_f$ and $L^2_f$, since the initial variable valuations $\vecinit^1$ and $\vecinit^2$ are fixed. On the other hand, upon symbolically evaluating the expected values appearing in constraints, all other constraints collected in Step~2 above are of the form
\begin{equation}\label{eq:handelmaneq}
	\forall \mathbf{x} \in \mathbb{R}^{|V|}.\, \textit{lin-exp}_1(\mathbf{x}) \geq 0 \land \dots \land \textit{lin-exp}_k(\mathbf{x}) \geq 0 \Rightarrow \textit{poly-exp}(\mathbf{x}) \geq 0,
\end{equation}
where $V\in\{V^1,V^2\}$, $\textit{lin-exp}_i$ is a linear expression over variables in $V$ for each $1 \leq i \leq k$, and $\textit{poly-exp}$ is a polynomial expression over variables in $V$ (equalities and absolute values are encoded by two inequality constraints). This is due to our algorithm assumptions that the supporting invariants and transition guards are all defined in terms of linear expressions. Furthermore, the linear coefficients in each $\textit{lin-exp}_i$ are constant values determined by transition guards or by supporting invariants, hence they do not contain any symbolic template variables.

It was shown in~\cite{ChatterjeeFG16,Wang0GCQS19,ZikelicCBR22,AsadiC0GM21} that entailments as in eq.~\eqref{eq:handelmaneq} can be translated into purely existentially quantified {\em linear constraints} over the symbolic template variables (and auxiliary variables introduced by the translation), by using Handelman's theorem~\cite{handelman1988representing}. 
Using this translation, we convert the system of constraints collected in Step~2 into a system of purely existentially quantified linear constraints over the symbolic template variables and auxiliary variables introduced in translation. Thus, we obtain a linear programming (LP) instance without an optimization objective.

\paragraph{Step~4: LP solving.} We feed the resulting LP instance to an off-the-shelf LP solver. The algorithm returns ``Not output-equivalent'' and outputs the computed $f$, $U^1_f$ and $L^2_f$ if the LP is successfully solved, or returns ``Unknown'' otherwise.

\medskip The following theorem establishes soundness of our algorithm for the equivalence refutation analysis. The proof can be found in Section~\ref{app:algoproofs} in the supplementary material.

\begin{restatable}[Correctness of Equivalence Refutation]{theorem}{algosoundness}\label{thm:algosoundness}
	Suppose that the algorithm outputs ``Not output-equivalent''. Then $\pCFG_1$ and $\pCFG_2$ are indeed not output-equivalent, and $U^1_f$ and $L^2_f$ are valid UESM and LESM for $f$, respectively.
\end{restatable}

\subsection{Algorithm for the Similarity Refutation Problem}\label{sec:algosimil}

We now outline the key additional steps needed to extend our algorithm in Section~\ref{sec:algoequiv} to an algorithm for the Similarity refutation problem. For the interest of space, we omit the details and defer them to Section~\ref{app:similarityalgo} in the Supplementary material.

In addition to the parameters listed in Section~\ref{sec:algoequiv}, our algorithm for the Similarity refutation problem is also parameterized by the choice of a metric $d$ over outputs and a lower bound $\epsilon > 0$ on the Kantorovich distance that we wish to prove. We allow any of the following standard metrics: $L^p$-metric, discrete metric, and uniform metric (all defined in Section~\ref{app:similarityalgo}). The rest of the algorithm proceeds analogously as in Section~\ref{sec:algoequiv}, with the only difference being that in Step~2 of the algorithm we need to collect two additional constraints: (1)~relational constraint on the lower bound on Kantorovich distance, and (2)~$1$-Lipschitz continuity of the function $f$ on outputs.

\smallskip\noindent{\bf\em Optimization of the Kantorovich distance.} We note that our algorithm for the Similarity refutation problem reduces the synthesis of $f$, $U^1_f$ and $L^2_f$ to an LP instance without the optimization objective. Thus, our method can also {\em optimize} the lower bound on the Kantorovich distance by treating \( \epsilon \) as a variable and adding the optimization objective to maximize $\epsilon$.

\section{Experimental Results}\label{sec:exper}

We implemented a prototype\footnote{We will make our prototype tool publicly available. Link to the implementation hidden for double blind reviewing.} of our methods for equivalence and similarity refutation (the latter in terms of Kantorovich distance with respect to the \( L^1 \) metric). Our prototype takes as input two probabilistic programs having the same set of output variables $\Vout$. The tool then checks (i) whether the output distributions of two programs are equivalent, and if not, (ii) whether it can compute a lower bound on their Kantorovich distance.

\smallskip\noindent{\bf\em Implementation.} We implemented our prototype in Java. We used Gurobi \cite{gurobi} to solve LP instances and ASPIC \cite{FeautrierG10} and STING \cite{SankaranarayananSM04} to generate supporting linear invariants. Our implementation uses rational numbers for storing coefficients and variable values to avoid rounding and floating-point errors. For each input program pair, our prototype attempts to synthesize UESM/LESMs of a varying polynomial degree ranging from~$1$~to~$5$, progressively increased in case of failure. All experiments were run on an Ubuntu 22.04 machine with an 11th Gen Intel Core i5 CPU and 16 GB RAM with a timeout of 10 minutes. 

\begin{table}[b]
	\caption{Experimental results showing: (1) comparison of our equivalence refutation method and the baseline, (2) lower boundson Kantorovich distance computed by our similarity refutation method, and (3) time taken to solve each instance. A $\checkmark$ in the ``Eq. Ref.'' column represents that the tool successfully refuted equivalence of the two input programs, ``TO'' and ``NA'' stand for ``timeout'' and ``Not Applicable'', respectively.}
	\vspace{-1em}
	\centering
	\texttt{
		\resizebox{\textwidth}{!}
		{
			\fontsize{4.9pt}{4.9pt}\selectfont
	\begin{tabular}[h]{|c|c c|c c c c|c c|}
		\hline
	& \multirow{2}{*}{Name}& \multirow{2}{*}{Loop} & \multicolumn{4}{c|}{Our Method} & \multicolumn{2}{c|}{PSI + Mathematica}\\
		& & & Eq. Ref. & Time(s) & Distance & Time(s) & Eq. Ref. & Time(s) \\
		\hline 
	\multirow{40}{*}{\rotatebox[origin=c]{90}{Benchmarks from \cite{Wang0GCQS19}}} & \multirow{4}{*}{Simple Example} 
	& 10 & \checkmark & 0.74 & 6.667 & 0.76 & \checkmark & 6.30 \\
	& & 100 & \checkmark & 0.41 & 66.667 & 0.40 & TO & - \\
	& & 1000 & \checkmark & 0.34 & 666.667 & 0.34 & TO & - \\
	& & original & \checkmark & 0.30 & 266.667 & 0.25 & NA & - \\
	\cline{2-9}	
	& \multirow{4}{*}{Nested Loop}
	& 10 & \checkmark & 5.31 & 1.667 & 5.13 & TO & - \\
	& & 100 & \checkmark & 3.30 & 16.667 & 3.13 &TO & -\\
	& & 1000 & \checkmark & 2.84 & 166.667 & 2.82 & TO & - \\
	& & original & \checkmark & 0.31 & 50 & 0.33 & NA & - \\
	\cline{2-9}	
	& \multirow{4}{*}{Random Walk} 
	& 10 & \checkmark & 0.76 & 2 & 0.69 & \checkmark & 3.37\\
	& & 100  & \checkmark & 0.33 & 20 & 0.28 & TO & - \\
	& & 1000 & \checkmark & 0.27 & 200 & 0.29 &TO& - \\
	& & original & \checkmark & 0.24 & 9 & 0.22 & NA & - \\
	\cline{2-9}	
	& \multirow{4}{*}{Goods Discount} 
	& 10 & \checkmark & 0.88 & 0.125 & 0.46 & TO & -\\
	& & 100 & \checkmark & 0.36 & 0.125 & 0.56 & TO & -\\
	& & 1000 & \checkmark & 0.30 & 0.125 & 0.53 & TO & -\\
	& & original & \checkmark & 0.35 & 0.008 & 0.56 & NA & - \\
	\cline{2-9}	
	& \multirow{4}{*}{Pollutant Disposal} 
	& 10 & \checkmark & 2.68 & 3.272 & 2.76 & TO & -\\
	& & 100 & \checkmark & 3.05 & 32.727 & 2.83 & TO & -\\
	& & 1000  & \checkmark & 3.34 & 327.272 & 3.18 & TO & -\\
	& & original & \checkmark & 0.44 & 0.026 & 0.46 & NA & - \\
	\cline{2-9}	
	& \multirow{4}{*}{2D Robot} 
	& 10 & \checkmark & 0.59 & TO & - & TO & -\\
	& & 100 & \checkmark & 0.57 & TO & - & TO & -\\
	& & 1000 & \checkmark & 0.62 & TO & - & TO & -\\
	& & original & \checkmark & 13.90 & TO & - & NA & - \\
	\cline{2-9}	
	& \multirow{4}{*}{Bitcoin Mining} 
	& 10 & \checkmark & 0.40 & 0.005 & 0.46 & \checkmark & 1.22\\
	& & 100 & \checkmark & 0.23 & 0.05 & 0.21 & \checkmark & 34.79\\
	& & 1000 & \checkmark & 0.34 & 0.5 & 0.46 & TO & -\\
	& & original & \checkmark & 0.25 & 0.05 & 0.22 & TO & - \\
	\cline{2-9}	
	& \multirow{4}{*}{Bitcoin Mining Pool} 
	& 10 & \checkmark & 205.51 & 225.25 & 140.7 & \checkmark & 1.90\\
	& & 100 & \checkmark & 121.43 & 22525 & 124.17 & TO & -\\
	& & 1000 & \checkmark & 235.57 & 2252500 & 258.49 & TO & -\\
	& & original & \checkmark & 129.98 & 122761.25 & 131.06 & NA & - \\
	\cline{2-9}	
	& \multirow{4}{*}{Species Fight} 
	& 10 & TO & - & TO & - & TO & -\\
	& & 100 & TO & - & TO & - & TO & -\\
	& & 1000 & TO & - & TO & - & TO & - \\
	& & original & \checkmark & 0.90 & TO & - & NA & - \\
	\cline{2-9}	
	& \multirow{4}{*}{Queuing Network} 
	& 10 & \checkmark & 0.99 & TO & - & TO & 94.42\\
	& & 100 & \checkmark & 0.77 & TO & - & TO & -\\
	& & 1000 & \checkmark & 0.81 & TO & - & TO & -\\
	& & original & \checkmark & 0.79 & TO & - & TO & - \\
	\cline{1-9}
	\multirow{28}{*}{\rotatebox[origin=c]{90}{Benchmarks from \cite{DBLP:conf/tacas/KuraUH19}}} & \multirow{4}{*}{coupon\_collector} 
	& 10 & \checkmark & 0.67 & 0.167 & 0.89 & \checkmark & 1.21\\
	& & 100 & \checkmark & 0.64 & 0.458 & 0.85 & TO & 6.45\\
	& & 1000 & \checkmark & 1.45 & 0.496 & 2.41 & TO & -\\
	& & original & \checkmark & 1.10 & 0.5 & 1.41 & NA & - \\
	\cline{2-9}	
	& \multirow{4}{*}{coupon\_collector4} 
	& 10 & \checkmark & 17.28 & 0.216 & 19.83 & TO & 3.04\\
	& & 100 & \checkmark & 16.92 & 1.215 & 18.49 & TO & -\\
	& & 1000 & \checkmark & 31.18 & 1.471 & 27.93 & TO & -\\
	& & original & \checkmark & 70.96 & TO & TO & NA & - \\
	\cline{2-9}	
	& \multirow{4}{*}{random\_walk\_1d\_intvalued} 
	& 10 & \checkmark & 0.25 & 2 & 0.34 & \checkmark & 1.49\\
	& & 100 & \checkmark & 0.20 & 20 & 0.19 & \checkmark & 6.24\\
	& & 1000 & \checkmark & 0.41 & 200 & 0.41 & TO & -\\
	& & original & \checkmark & 0.32 & 1.2 & 0.38 & NA & - \\
	\cline{2-9}	
	& \multirow{4}{*}{random\_walk\_1d realvalued} 
	& 10 & \checkmark & 0.29 & 2 & 0.33 & TO & -\\
	& & 100 & \checkmark & 0.22 & 20 & 0.28 & TO & -\\
	& & 1000 & \checkmark & 0.48 & 200 & 0.54 & TO & -\\
	& & original & \checkmark & 0.27 & 3.841 & 0.50 & NA & - \\
	\cline{2-9}	
	& \multirow{4}{*}{random\_walk\_1d\_adversary} 
	& 10 & \checkmark & 0.28 & 12.425 & 0.22 & \checkmark & 1.49\\
	& & 100 & \checkmark & 0.26 & 138.425 & 0.22 & TO & 91.59\\
	& & 1000 & \checkmark & 0.59 & 1398.425 & 0.56 & TO & -\\
	& & original & \checkmark & 0.38 & 0.768 & 0.40 & NA & - \\
	\cline{2-9}	
	& \multirow{4}{*}{random\_walk\_2d\_demonic} 
	& 10 & \checkmark & 0.25 & 2 & 0.34 & TO & -\\
	& & 100 & \checkmark & 0.27 & 20 & 0.27 & TO & -\\
	& & 1000 & \checkmark & 0.65 & 200 & 0.67 & TO & -\\
	& & original & \checkmark & 0.58 & 0.668 & 0.58 & NA & -\\
	\cline{2-9}	
	& \multirow{4}{*}{random\_walk\_2d\_variant} 
	& 10 & \checkmark & 0.40 & 2 & 0.37 & TO & -\\
	& & 100 & \checkmark & 0.31 & 20 & 0.37 & TO & -\\
	& & 1000 & \checkmark & 0.83 & 200 & 0.94 & TO & -\\
	& & original & \checkmark & 0.75 & 0.501 & 0.76 & NA & - \\
	\cline{1-9}	
	\multirow{4}{*}{\rotatebox[origin=c]{90}{Sec. \ref{sec:overview}}} & \multirow{4}{*}{Transmission protocol (Figure~\ref{fig:motivatingtwo})} 
	& 10 & \checkmark & 0.27 & 0.005 & 0.24 & \checkmark & 1.19 \\
	& & 100 & \checkmark & 0.19 & 0.05 & 0.21 & \checkmark & 12.23 \\
	& & 1000 & \checkmark & 0.47 & 0.5 & 0.57 & TO & - \\
	& & original & \checkmark & 0.21 & 1.001 & 0.19 & TO & - \\
	\cline{2-9}
	\cline{1-9}
	\multicolumn{3}{|c|}{Count } & 69 & - & 59 & - & 15 & - \\
	\multicolumn{3}{|c|}{Average } & - & 12.88 & - & 12.94 & - & 17.77 \\
	\cline{1-9}
	\end{tabular}
	}
	}
\label{tab:experimeriments-finite-loop}
\end{table}

\smallskip\noindent{\bf\em Baseline.} To refute the equivalence of two probabilistic programs, an alternative approach would be to use a state of the art symbolic integration tool such as~PSI \cite{DBLP:conf/cav/GehrMV16} to first compute probability density functions of output distributions of two probabilistic programs, and then check whether the two density functions are identical. Hence, we compare our method against a baseline which first uses PSI~\cite{DBLP:conf/cav/GehrMV16} to compute the probability density functions and then uses Mathematica~ \cite{Mathematica} to compare the density functions. We note that PSI (and so our baseline) is only applicable to programs with statically bounded loops.

\smallskip\noindent{\bf\em Benchmarks.} As our benchmark set, we consider loopy probabilistic programs collected from \cite{Wang0GCQS19,DBLP:conf/tacas/KuraUH19}. These benchmarks model various different applications, ranging from classical examples
of random walks and their variants, coupon collector, to examples of academic interest such as queueing network and species fight, to realistic examples including Bit-coin mining. These programs have been verified to have finite expected termination time and bounded variable updates, which are sufficient conditions for the satisfaction of the (C4) OST-soundness condition as discussed in Section~\ref{sec:algo}. Programs with statically bounded loops also satisfy the (C1) OST-soundness condition. As some of these programs contain unbounded loops which are not supported by PSI and so by our baseline, for each collected program we also consider three modifications where we force the loops to terminate after at most 10, 100 and 1000 iterations, respectively.

Each collected program was used to construct a pair of programs for equivalence and similarity analysis as follows: if the program contained non-deterministic branching, we constructed two programs of which one always chooses the if-branch and the other chooses the else-branch of the non-deterministic choice. For programs without non-determinism, we obtain the second program by injecting a small perturbation into exactly one sampling instruction, without further changes.


\smallskip\noindent{\bf\em Discussion of Results.} Table \ref{tab:experimeriments-finite-loop} shows our experimental results on the benchmark set described above together with the illustrating example from Section \ref{sec:overview}. Our method shows much better scalability compared to the baseline as the loop bound parameter is increased, demonstrating the advantage of a static analysis method that does not rely on (symbolic) execution of probabilistic programs with long executions. Moreover, our method is also able to compute Kantorovich distance lower bounds for most benchmarks. To the best of our knowledge, our method is the first automated method to compute lower bounds on Kantorovich distance. 


As mentioned above, for the purpose of our experiments in Table \ref{tab:experimeriments-finite-loop}, we used the (C4) OST-soundness condition which has been verified to be satisfied by all benchmarks. However, conditions (C1)-(C3) are also applicable to many of our benchmarks. An experimental comparison of the performance of our tool with different OST-soundness conditions is provided in Section \ref{sec:experiment-comparison-C1-C2-C3} in the Supplementary material. 

We also observe one practical limitation of our approach: the performance of our automated method is dependent on the quality of {\em supporting linear invariants} generated for both programs. For some benchmarks in Table~\ref{tab:experimeriments-finite-loop} (e.g. \texttt{coupon\_collector}) the computed lower bound on distance does not scale with the number of loop iterations. We believe this is due to linear invariants generated for these programs being imprecise. When more precise invariants are available (e.g. \texttt{Nested Loop}), our method derives tighter bounds on distance. Moreover, while linear invariant generation is very efficient in most of our benchmarks, in some cases (e.g.~\texttt{Bitcoin Mining Pool}) this was a highly computationally expensive task, leading to larger runtimes of our tool. There are also cases like the \texttt{2D robot} benchmark, where STING times out without returning any invariants, which is why our tool fails to compute a distance lower-bound. However, the invariants returned by ASPIC are strong enough for our tool to disprove equivalence. In other benchmarks where equivalence is refuted but no lower bound on distance is computed, we observe that supporting invariants are unbounded and so our tool could not normalize the function $f$ on outputs to make it $1$-Lipschitz continuous. Lastly, in cases like the finite loop instances of the \texttt{Species Fight} benchmark, non-polynomial functions are required for disproving equivalence, thus our tool fails. 


\smallskip\noindent{\bf\em Summary of Results.} Our experiments demonstrate that our automated method can refute equivalence and compute lower bounds on the Kantorovich distance for a wide variety of probabilistic programs (Table \ref{tab:experimeriments-finite-loop}). These results highlight scalability and efficiency of the method, especially when compared to the baseline based on symbolic integration. Hence, while suffering from certain limitations discussed above, we conclude that our method is applicable to a wide range of programs.

\vspace{-0.5em}
\section{Related Work}\label{sec:relatedwork}

We discuss many of the existing static analyses for single probabilistic programs and statistical testing techniques for probability distributions in Section~\ref{sec:intro}. Moreover, we provide a discussion on the comparison of our UESMs and LESMs to cost supermartingales of~\cite{Wang0GCQS19} and super- and sub-invariants of~\cite{HarkKGK20} in Remark~\ref{rmk:comparison} and in Section~\ref{sec:expectationsupermartingales}. Hence, we omit repetition and in the rest of this section we overview some other prior works on relational analysis of (probabilistic) programs.

\smallskip\noindent{\bf\em Sensitivity analysis.} Sensitivity analysis is a relational property that has received a lot of attention in static analysis of probabilistic programs~\cite{BartheEGHS18,0001BHKKM21,WangFCDX20,HuangWM18}. Given a probabilistic program and two inputs, the goal of sensitivity analysis is to derive bounds on the distance between output distributions on those inputs, towards verifying e.g.~differential privacy~\cite{AlbarghouthiH18,BartheGGHS16}. A prominent method for sensitivity analysis in probabilistic programs is based on coupling proofs~\cite{BartheEGHS18,AlbarghouthiH18,BartheGGHS16}. While sensitivity analysis considers two inputs given the {\em same} probabilistic program, in this work we study the equivalence and similarity refutation problems for {\em probabilistic program pairs}. Our method does not assume any level of syntactic similarity of control flows of the two programs. In contrast, sensitivity analysis and methods based on coupling proofs often exploit the ``aligned'' control flow of two executions on sufficiently close inputs.


\smallskip\noindent{\bf\em Equivalence analysis for finite-state probabilistic programs.} There is a significant body of work on comparing~\emph{finite-state} Markov chains and Markov decision processes w.r.t. various notions of equivalence and similarity. This includes computing the total variation distance~\cite{DBLP:conf/csl/ChenK14,DBLP:conf/icalp/Kiefer18}, trace equivalence~\cite{DBLP:conf/fsttcs/Kiefer020}, contextual equivalence~\cite{DBLP:conf/tacas/LegayMOW08}, and probabilistic bisimilarity~\cite{LARSEN19911, TracolDZ11}. These works focus on finite-state models, whereas we consider Turing-complete probabilistic programs.

\smallskip\noindent{\bf\em Accuracy of probabilistic inference.} Several works have studied accuracy of probabilistic inference algorithms, e.g. by considering  auxiliary inference divergence \cite{Cusumano-Towner17}, bidirectional Monte Carlo \cite{GrosseGA15,GrosseAR16} and symmetric divergence over simulations \cite{Domke2021}. The key distinction between these works and ours is that the former provide statistical guarantees, such as bounds for Kullback-Leibler (KL) divergence in expectation, while our work provides provably valid guarantees via static analysis.


\smallskip\noindent{\bf\em Relational analyses in non-probabilistic programs.} Prior work has studied static analysis with respect to a number of relational properties in non-probabilistic programs, including equivalence proving~\cite{GodlinS13,FelsingGKRU14}, semantic differencing~\cite{PartushY13,PartushY14}, continuity~\cite{ChaudhuriGL10} or differential cost analysis~\cite{CicekBG0H17,QuGG21,ZikelicCBR22}. In particular, the method of~\cite{ZikelicCBR22} computes a bound on the difference in cost usage of a program pair by simultaneously computing an upper bound on cost for one program and a lower bound on cost for the other program, similarly to what we do with the synthesis of UESMs and LESMs. However, the key algorithmic difference is that we also {\em simultaneously synthesize the function on outputs} with respect to which the UESM and the LESM are defined. In contrast, in differential cost analysis a cost function is known a priori.



\section{Conclusion}\label{sec:conclu}

We presented a new martingale-based method for refuting the equivalence and similarity of output distributions of probabilistic programs. Our approach is fully automated, applicable to infinite-state programs, and provides formal guarantees on the correctness of its result. Our experimental results demonstrate the effectiveness of our approach on a range of probabilistic program pairs. 

An interesting direction for future work is the extension of our methods to probabilistic programs with \emph{observe} statements~\cite{DBLP:journals/toplas/OlmedoGJKKM18}, that can express \emph{epistemic} uncertainty about the system modeled by the program. The observe statements condition the output distribution by the event that all observations along the run are satisfied, and the task would be to refute the equivalence and similarity of such conditional distributions. Another direction is to consider improving our method for similarity refutation for programs with unbounded outputs, as discussed in Section~\ref{sec:exper}. Yet another direction is to study the equivalence and similarity refutation problems for probabilistic programs that do not terminate almost-surely. Such programs define sub-distributions over their outputs, hence methods for these problems would need to reason about pairs of sub-distributions.

\section{Acknowledgements}
This research was partially supported by the ERC CoG 863818 (ForM-SMArt) grant. Petr Novotný is supported by the Czech Science Foundation grant no. GA23-06963S.

\bibliographystyle{ACM-Reference-Format}
\bibliography{bibliography}

\newpage
\appendix

\begin{center}
\Large Supplementary Material
\end{center}

    
\section{Further Motivating Examples}\label{sec:further-motive}

In this section, we present two more motivating examples for the equivalence and similarity refutation problems in probabilistic programs, whose analysis requires new approaches. The first example illustrates a compilation bug for probabilistic programs containing normal distributions, hence giving rise to infinite-state probabilistic programs. Finally, the second example shows two probabilistic programs that are syntactically similar and only slightly differ in probability distributions appearing in their sampling instructions, for which the equivalence and similarity refutation are quite challenging.

\begin{figure}[h]
    \centering
    \begin{subfigure}{0.47\textwidth}
    \begin{lstlisting}[frame=none,numbers=none,escapechar=@,mathescape=true]
            $\,\texttt{i} = 1,\, \texttt{n} = 10\,000,\, \texttt{sum} = 0$
        $\locinit$:   while $\texttt{i} \leq \texttt{n}$:
        $\loc_1$:       $\texttt{r} = \Normal(0,1)$
        $\loc_2$:       $\texttt{sum} = \texttt{sum} + \texttt{r}$
        $\loc_3$:       $\texttt{i} = \texttt{i} + 1$
        $\locterm$: return $\texttt{sum}$
    \end{lstlisting}
    \end{subfigure}
    \hfill
    \begin{subfigure}{0.47\textwidth}
    \begin{lstlisting}[frame=none,numbers=none,escapechar=@,mathescape=true]
            $\,\texttt{n} = 10\,000,\, \texttt{sum} = 0$
        $\locinit$:   $\texttt{r} = \Normal(0,1)$
        $\loc_1$:   $\,\texttt{sum} = \texttt{sum} + \texttt{n} \cdot \texttt{r}$
    
    
        $\locterm$: return $\texttt{sum}$
    \end{lstlisting}
    \end{subfigure}
    \vspace{-0.5em}
    \caption{Compilation bug example.}
    \vspace{-1em}
    \label{fig:motivating-compile}
    \end{figure}

    \begin{example}[Compilation bug detection]
        Consider the probabilistic program in Figure~\ref{fig:motivating-compile} left. It initializes the program variable $\texttt{sum}$ to $0$, iteratively and independently samples $\texttt{n} = 10\,000$ values from the standard normal distribution and adds the sampled values to $\texttt{sum}$. Upon termination, the program returns the value of $\texttt{sum}$. Thus, the output distribution of this program is the probability distribution of $\texttt{sum}$ upon termination.
        
        In order to present an instance of a compilation bug, suppose now that a compiler for optimization replaces the loop which repeatedly samples and adds identically distributed random values to $\texttt{sum}$ with code that samples only a single value from the standard normal distribution, multiplies it by $\texttt{n} = 10\,000$ and adds the result to $\texttt{sum}$, giving rise to the program in Figure~\ref{fig:motivating-compile} right. Hence, instead of repeatedly sampling values and adding them to $\texttt{sum}$, we only need to sample one value.
        
        While in deterministic programs this would be a sound optimization, these two probabilistic programs do not produce equivalent output distributions. The reason behind inequivalence is subtle -- this optimization is agnostic to the fact that samples in Figure~\ref{fig:motivating-compile} left are {\em independent}. Indeed, since the sum of two independent normal random variables distributed according to $\Normal(\mu_1,\sigma_1^2)$ and $\Normal(\mu_2,\sigma_2^2)$ is distributed according to $\Normal(\mu_1+\mu_2,\sigma_1^2+\sigma_2^2)$, it follows that the value of $\texttt{sum}$ upon termination of the program in Figure~\ref{fig:motivating-compile} left is distributed according to $\Normal(0,10\,000)$. On the other hand, the value of $\texttt{sum}$ upon termination of the program in Figure~\ref{fig:motivating-compile} right is distributed according to $10\,000 \cdot \Normal(0,1) = \Normal(0,10\,000 \cdot 10\,000)$. Hence, these two output distributions are not equivalent. However, since normal distribution has infinite support, we cannot use methods for finite-state probabilistic programs to refute equivalence in this example.
    \end{example}

    \begin{figure}[h]
    \centering
    \begin{subfigure}{0.49\textwidth}
    \begin{lstlisting}[frame=none,numbers=none,escapechar=@,mathescape=true]
        $\,\,l_1 = 0,\, l_2 = 0,\, i = 1,\, n = 10\,000,\, \texttt{time} = 0$
    $\locinit$: while $i \leq n$:
    $\loc_1$:     if $l_1 \geq 1$ then $l_1 = l_1 - 1$ fi
    $\loc_2$:     if $l_2 \geq 1$ then $l_2 = l_2 - 1$ fi
    $\loc_3$:     if prob($0.02$) then:
    $\loc_4$:        if prob($0.2$) then:
    $\loc_5$:             $l_1 = l_1 + 3$
    $\loc_6$:        elif prob($0.5$) then:
    $\loc_7$:             $l_2 = l_2 + 2$
    $\loc_8$:        else:
    $\loc_9$:             $l_1 = l_1 + 2,\, l_2 = l_2 + 1$
    $\loc_{10}$:     if $l_1 \geq l_2$ then:
    $\loc_{11}$:          $\texttt{time} = \texttt{time} + l_1$
    $\loc_{12}$:     else:
    $\loc_{13}$:          $\texttt{time} = \texttt{time} + l_2$
    $\loc_{14}$:     $i = i + 1$
    $\locterm$: return $l_1,\,l_2,\,\texttt{time}$
    \end{lstlisting}
    \end{subfigure}
    \hfill
    \begin{subfigure}{0.49\textwidth}
    \begin{lstlisting}[frame=none,numbers=none,escapechar=@,mathescape=true]
        $\,\,l_1 = 0,\, l_2 = 0,\, i = 1,\, n = 10\,000,\, \texttt{time} = 0$
    $\locinit$: while $i \leq n$:
    $\loc_1$:     if $l_1 \geq 1$ then $l_1 = l_1 - 1$ fi
    $\loc_2$:     if $l_2 \geq 1$ then $l_2 = l_2 - 1$ fi
    $\loc_3$:     if prob($0.02$) then:
    $\loc_4$:       if prob($0.15$) then:
    $\loc_5$:           $l_1 = l_1 + 3$
    $\loc_6$:       elif prob($0.45$) then:
    $\loc_7$:           $l_2 = l_2 + 2$
    $\loc_8$:       else:
    $\loc_{9}$:         $l_1 = l_1 + 2,\, l_2 = l_2 + 1$
    $\loc_{10}$:    if $l_1 \geq l_2$ then:
    $\loc_{11}$:        $\texttt{time} = \texttt{time} + l_1$
    $\loc_{12}$:    else:
    $\loc_{13}$:        $\texttt{time} = \texttt{time} + l_2$
    $\loc_{14}$:    $i = i + 1$
    $\locterm$: return $l_1,\,l_2,\,\texttt{time}$
    \end{lstlisting}
    \end{subfigure}
    \vspace{-0.5em}
    \caption{The Fork and Join queuing network example.}
    \vspace{-1em}
    \label{fig:motivating-queue}
    \end{figure}

    \begin{example}[Refuting similarity of syntactically similar programs]
	Consider the probabilistic program pair in Figure~\ref{fig:motivating-queue}. Each program models a Fork and Join (FJ) queuing network with $2$ processors, each with its own queue~\cite{AlomariM14,Wang0GCQS19}. Both programs model processes that evolve over $\texttt{n} = 10\,000$ time steps, and program variables $l_1$ and $l_2$ denote the queue lengths. At each time step, one job unit is processed by each queue, thus the length of each queue is decreased by $1$. However, with probability $0.02$ new jobs may arrive. The FJ network then probabilistically decides whether to assign the job to one queue or to divide it between two queues. Jobs are assumed to be identical. It takes $3$ units of time for the first queue and $2$ units of time for the second queue to complete the job alone. If the job is divided between the queues, then they take $2$ and $1$ units of time to complete their part of the job, respectively.
	
	In the program in Figure~\ref{fig:motivating-queue} left which is taken from~\cite{Wang0GCQS19}, a job is assigned to the first queue with probability $0.2$, to the second queue with probability $0.8 \cdot 0.5 = 0.4$, and is divided between the two queues with the remaining probability. In the program in Figure~\ref{fig:motivating-queue} right, we slightly decrease the probabilities of assigning a job to individual queues and increase the probability of dividing the job between the queues. In particular, a job is now assigned to the first queue with probability $0.15$, to the second queue with probability $0.85 \cdot 0.45 = 0.3825$, and is divided between the queues with the remaining probability. In both programs, program variable $\texttt{time}$ models the total processing time of all jobs. Note that the total processing time is computed from the perspective of each job -- it also accounts for the waiting times for jobs already in the queue to be solved first. For each job, the processing time is defined by the length of the longest queue at the time of the job addition~\cite{Wang0GCQS19}. Both programs output variables $l_1$, $l_2$ and $\texttt{time}$ upon termination. Hence, the output distribution of each program is the joint probability distribution of the values of these three program variables upon termination
	
	Note that the difference between these two probabilistic programs is quite subtle. We do not decrease the probability of assigning a job to the slower processor $l_1$ while increasing the probability of assigning it to the faster processor $l_2$, or vice-versa. Rather, we decrease both probabilities and simply increase the probability of the job being divided between the two queues. Thus, since no queue is preferred by this change and since changes in probabilities are small (recall that jobs arrive only with probability $0.02$), at first glance it is not clear how close are the output distributions of these two programs. The problem of refuting equivalence or computing lower bounds on Kantorovich distance between these two output distributions is highly challenging both for static analysis (due to syntactic similarity) and for statistical testing (due to long execution times).
\end{example}

\section{Finite First Moments}
\label{app:moments}
We discuss several sufficient conditions for the finite first moment assumption to be satisfied and methods through which the assumption can be enforced at the cost of modifying the compared programs in a principled way.

First, if a metric \( d \) is bounded over the output space, then all distributions over the output space have finite first moments w.r.t. \( d \). An important example of such a metric is the discrete metric, meaning that the similarity refutation problem w.r.t. the total variation distance can be considered for any pair of programs. This is in line with our previous observation that the total variation distance does not impose any restriction on the compared distributions.

If the metric is not \emph{a priori} bounded, the finite first moment assumption is satisfied as long as the \emph{ranges} of both programs (i.e. the subset of \( \Rset^{|V_\out|} \) containing exactly the possible outputs of the programs) \emph{bounded} w.r.t. \( d \). Formally, for a pCFG \( \pCFG \) we have
\[\mathit{range}(\pCFG) = \{\val^\out \in \Rset^{|V_\out|} \mid \exists \rho \in \Run^\pCFG \text{ that reaches a terminal state } (\loc_\out, \val)\},   \]
and \( \pCFG \) has a bounded range if \( \sup\{d(\veca{x},\veca{y}) \mid \veca{x}, \veca{y} \in \mathit{range}(\pCFG)\} < \infty \). For the standard \( L \)-metrics, the \( \mathit{range}(\pCFG) \) is bounded iff it is contained in some bounded \( |V_\out| \)-dimensional hyperrectangle; this can be checked e.g. by computing a (non-probabilistic) inductive invariant of the program and investigating the shape of the invariant in the location \( \loc_\out \), see Section~\ref{sec:algo}.

Another way to ensure finite first moment w.r.t. metric \( d \) is to show that the output distribution has \emph{exponentially decreasing} tails, in the sense that there is some \( \veca{x}_0 \in \Rset^{|V_\out|} \) s.t. the probability of outputting an element of \( d \)-distance larger than \( \gamma \) from \( \veca{x} \) decreases to zero exponentially fast as \( \gamma \) increases to infinity. We are not aware of any automated method tailor-made for proving this property of output distributions, though exponentially decreasing tails of other characteristics of probabilistic programs (such as termination time) were studied before.~\cite{DBLP:conf/tacas/KuraUH19}. However, exponentially decreasing tails of the output distribution can be sometimes inferred manually or provided as form of domain knowledge: for instance, in our running example~\ref{fig:motivatingtwo} it is easy to see that the outputs follow a normal distribution, which is well known to have exponentially decreasing tails.

If none of the above is applicable, we can force finite first moments by artificially ``clipping'' the range of the programs involved into the (same) bounded set. That is, the user can fix, e.g. a hyperrectangle \( [n_1,m_1]\times[n_2,m_2]\times\cdots \times [n_{|V_\out|}]\times [m_{|V_\out|}] \), and instrument each of the two programs so that upon termination, the value of each output variable \( x_i \) is clipped into the interval \( [n_i,m_i] \). This of course does not solve the similarity refutation problem for the original programs, since the clipping alters the output distributions. However, the distributions are only altered outside of the interior of the hyperrectangle; hence, a lower bound on the distance of the clipped distributions is still a valid certificate of semantic difference of the \emph{original} programs, where the difference manifests itself inside the selected hyperrectangle.

\section{Probability Theory}
\label{app:probt}

A \emph{probability space} is a triple \( (\Omega, \sigmaAlg, \probm) \), where \( \Omega \) is a sample space, \( \sigmaAlg \) is a sigma-algebra over \( \Omega \) (a collection of subsets of \( \Omega \) containing \( \Omega \) and closed under complementation and countable unions) and \( \probm \colon \sigmaAlg \rightarrow [0,1] \) is a probability measure on \( \sigmaAlg \), i.e. a function such that (i) \( \probm(\Omega) = 1 \); (ii) \( \probm(\Omega \setminus A) = 1- \probm(A) \) for each \( A \in \sigmaAlg \), and (iii) \( \probm(\bigcup_{i=1}^{\infty} A_i) = \sum_{i=1}^{\infty} \probm(A_i) \) for each sequence of pairwise disjoint sets \( A_1,A_2,\ldots \in \sigmaAlg \). 

A \emph{random variable} in a probability space \(  (\Omega, \sigmaAlg, \probm)  \) is a function \( R \colon \Omega \rightarrow \Rset \cup \{\pm \infty\} \) such that for each \( x \in \Rset \) it holds \( \{\omega \in \Omega \mid R(\omega) \leq x \} \in \sigmaAlg\) (such functions are also called \emph{\( \sigmaAlg \)-measurable}). We denote by \( \E_\probm[R] \) the expected value of \( R \) in \(  (\Omega, \sigmaAlg, \probm)  \), which is defined in the standard way via Lebesgue integration with respect to the measure \( \probm \)~\cite{Williams91}. We drop the \( \probm \) from the subscript if the probability measure is clear from the context.  A \emph{random vector} is a vector whose every component is a random variable. A (discrete-time) \emph{stochastic process} is a sequence of random vectors over the same probability space.

\subsection{Preliminaries for the Optional Stopping Theorem}

A \emph{filtration} over a sigma-algebra \( \sigmaAlg \) is a sequence of sigma-algebras \( (\sigmaAlg_{i})_{i=0}^{\infty} \) such that for each \( i \geq 0 \) it holds \( \sigmaAlg_i \subseteq \sigmaAlg_{i+1} \subseteq \sigmaAlg \). A stochastic process \( (\veca{X}_i)_{i=0}^{\infty} \) over \( \sigmaAlg \) is \emph{adapted} to such a filtration if for every \( i \geq 0 \) it holds that each component of \( \veca{X}_i \) is \( \sigmaAlg_{i} \)-measurable. Intuitively, the filtration categorizes the sets in \( \sigmaAlg \) so that sets in \( \sigmaAlg_i \) represent the information available at time \( i \).

Let \( (\Omega, \sigmaAlg, \probm) \) be a probability space and \( X \) be a random variable in this space. For a sub-sigma algebra \( \sigmaAlg' \subseteq \sigmaAlg \) the \emph{conditional expectation of \( X \) given \( \sigmaAlg' \)} is an \( \sigmaAlg' \)-measurable random variable denoted \( \E[X\mid \sigmaAlg'] \) such that for every set \( A\in \sigmaAlg' \) it holds \( \E[X \cdot \indicator{A} ] = \E[\E[X\mid \sigmaAlg']\cdot \indicator{A}]\), where \( \indicator{A} \) is the indicator function of the set \( A \). There may generally be zero or multiple \( \sigmaAlg' \)-measurable variables satisfying the defining condition of conditional expectation, but if at least one exists, all of the others differ from it only a set of zero probability. Hence, when at least one such random variable exists, any of them can be picked as \( \E[X\mid \sigmaAlg'] \).

If \( B \in \sigmaAlg\) is an event of positive probability, the conditional expectation of \( X \) given \( B \) is the expectation of \( X \) w.r.t. the probability measure \( \probm[~\cdot \mid B] = \frac{\probm[~\cdot~ \cap B]}{\probm[B]} \).

\begin{definition}
Let \( (Y_{i})_{i=0}^{\infty} \) be a (1-dimensional) stochastic process adapted to some filtration \( (\sigmaAlg_{i})_{i=0}^{\infty} \) such that \( \E[Y_{i+1} \mid \sigmaAlg_{i}] \) exists for every \( i \geq 0 \). We say that the process is a \emph{supermartingale} if for every \( i \geq 0 \) it holds \( \E[Y_{i+1}\mid \sigmaAlg_{i}] \leq Y_i \). We call the process a \emph{submartingale} if \( \E[Y_{i+1}\mid \sigmaAlg_{i}] \geq Y_i \) for every \( i \geq 0 \).
\end{definition}


\begin{definition}[Stopping time]
Let \( (\Omega, \sigmaAlg, \probm) \) be a probability space and \( (\sigmaAlg_{i})_{i=0}^{\infty} \) a filtration. A \emph{stopping time} is a random variable \( \stime \colon \Omega \rightarrow \Nset \cup \{\infty\} \) s.t.~\( \{ \omega \in \Omega \mid \stime(\omega) \leq t\} \in \sigmaAlg_t \) for any \( t \in \Nset \).
\end{definition}

\label{app:ost}

\begin{restatable}[Optional stopping theorem, OST]{theorem}{ostthm}\label{thm:ost}
Let \( (\Omega, \sigmaAlg, \probm) \) be a probability space. Next, let \( (Y_{i})_{i=0}^{\infty} \)  be a 1-dimensional stochastic process adapted to some filtration \( (\sigmaAlg_{i})_{i=0}^{\infty} \) of \( \sigmaAlg \), and \( \stime \) be a stopping time w.r.t. the same filtration \( (\sigmaAlg_{i})_{i=0}^{\infty} \). Assume that \( E[|Y_i|] < \infty \) for all \( i \geq 0 \) and that the above objects satisfy one of the following conditions:
\begin{itemize}
\item[(C1')] There exists a constant \( c \) such that \( \stime \leq c \) with probability 1 (i.e., the stopping time is almost-surely bounded).
\item[(C2')] There exists a constant \( c \) such that for each \( t \in \Nset \) and each \( \omega \in \Omega \) it holds \( |Y_{\min\{t, \stime(\omega)\}}(\omega)| \leq c \) (i.e., the process is bounded from both below and above up until the point of stopping).
\item[(C3')] \( \E[\stime] < \infty \), \( \E[|Y_0|] < \infty \), and there exists a constant \( c \) such that for every \( t \in \Nset \) it holds \( \E[|Y_{t+1}-Y_{t}|\mid \sigmaAlg_t] \leq c \) (i.e., the expected one-step change of the process is uniformly bounded over its evolution, even if conditioned by the whole past history of the process).
\end{itemize}

Then, \( \E[Y_\stime] \) is well-defined, and moreover \( \E[Y_\stime] \leq \E[Y_0]\) if \( (Y_{i})_{i=0}^{\infty} \) is a supermartingale and \( \E[Y_\stime] \geq \E[Y_0]\) if \( (Y_{i})_{i=0}^{\infty} \) is a submartingale. In other words, the expected value of, say supermartingale, at the point of stopping is bounded from above by its mean initial value, and dually for submartingales.
\end{restatable}

We conclude this section by restating the Extended optional stopping theorem from~\cite{Wang0GCQS19}.

\begin{theorem}[Extended OST,~\cite{Wang0GCQS19}]
Let \( (\Omega, \sigmaAlg, \probm) \) be a probability space. Next, let \( (Y_{i})_{i=0}^{\infty} \)  be a 1-dimensional stochastic process adapted to some filtration \( (\sigmaAlg_{i})_{i=0}^{\infty} \) of \( \sigmaAlg \), and \( \stime \) be a stopping time w.r.t. the same filtration \( (\sigmaAlg_{i})_{i=0}^{\infty} \). Assume that the above objects satisfy the following condition:
\begin{itemize}
\item[(C4')] There exist real numbers \( M, c_1, c_2, d \) such that (i) for all sufficiently large \( n \in \Nset \) it holds \( \probm(\stime > n) \leq c_1 \cdot e^{-c_2 \cdot n} \); and (ii) for all \( t \in \Nset \) it holds \( |Y_{n+1} - Y_n| \leq M\cdot n^d \).
\end{itemize}

Then, \( \E[Y_\stime] \) is well-defined, \( \E[|Y_i|] < \infty \) for every \( i \geq 0 \) and moreover, \( \E[Y_\stime] \leq \E[Y_0]\) if \( (Y_{i})_{i=0}^{\infty} \) is a supermartingale and \( \E[Y_\stime] \geq \E[Y_0]\) if \( (Y_{i})_{i=0}^{\infty} \) is a submartingale. 
\end{theorem}

\section{Proof of Theorem~\ref{thm:ulesm-exp-soundness}}
\label{sec:soundproof}
\soundness*

\begin{proof}
We present the proof for \( L_f \), the proof for \( U_f \) is analogous. 

Recall that \( Z_n \) denotes the \( n \)-th state along a run of \( \pCFG \) and \( \veca{X}_n \) denotes the \( n \)-th valuation encountered along a run. Let us define a stochastic process \( Y = (Y_n)_{n=0}^{\infty} \) by putting \( Y_n := L_f(Z_n) + f(\veca{X}_n^\out)\). The process \( Y \) is clearly adapted to the canonical filtration \( (\sigmaAlg_{n})_{n=0}^\infty \). 

First, note that \( \E[Y_{i+1} \mid \sigmaAlg_{i}] \) exists and for any run \( \rho \) is defined as follows: let \( \tau=(\loc,\prob) \) be the unique transition enabled in \( Z_i(\rho) \). Then
\begin{equation}
\label{eq:preexp}
\E[Y_{i+1} \mid \sigmaAlg_{i}] (\rho) = \sum_{\ell'\in L}\prob(\ell')\cdot \E[L_f(\ell', \mathbf{N}) + f(\mathbf{N}^\out)) ],
\end{equation}
where \( \mathbf{N} = \nextv(\tau, \veca{X}_i(\rho)) \). This function is well-defined, since the expectation on the right-hand side of \eqref{eq:preexp} always exists by the OST-soundness assumption. In Section~\ref{app:cond-exp}, we prove that the function defined in this way indeed satisfies the definition of conditional expectation.

Next, we prove that the process \( Y \) is a submartingale. We can continue from~\eqref{eq:preexp} as follows:

\begin{align*}
\E[Y_{i+1} \mid \sigmaAlg_{i}] (\rho) &= \sum_{\ell'\in L}\prob(\ell')\cdot \E[L_f(\ell', \mathbf{N}) + f(\mathbf{N}^\out)) ] &&\text{(by~\eqref{eq:preexp})} \\
&\geq  L_f(Z_i(\rho)) + f(Z_i(\rho))&&\text{(by~\eqref{eq:lfexp})} \\
&=Y_i(\rho) && \text{(by the def. of \( Y_i \)),} 
\end{align*}
as required.

In what follows, we abbreviate \( \TimeTerm \) by \( T \).
Since \( (\pCFG, L_f, f) \) is OST-sound, the submartingale \( Y \) satisfies the assumptions of either the optional stopping theorem or its extended variant. It follows that
\begin{align}
\label{eq:lfone}
L_f(\locinit, \vecinit) + f(\vecinit^\out) = \E[Y_0] \leq \E[Y_{T}] = \E[f(\veca{X}_T^\out)] =  \E_{\veca{x} \sim \mu^{\pCFG}}[f(\veca{x}^\out)],
\end{align}
where the first equality follows from the definition of \( Y \), the second from the (extended) optional stopping theorem, the third from the fact that \( L_f \) is zero upon termination, and last one from the definition of \( \mu^\pCFG \).
\end{proof}

%

\section{Conditional Expectation for U/LESMs}
\label{app:cond-exp}

We argue that 

\begin{equation}
\label{eq:supp-1}
\E[Y_{i+1} \mid \sigmaAlg_{i}] (\rho) = \sum_{\ell'\in L}\prob(\ell')\cdot \E[L_f(\ell', \mathbf{N}) + f(\mathbf{N}^\out)) ],
\end{equation}
where \( \tau \) is the unique transition enabled in state \( Z_i(\rho) \).

For each \( i \geq 0 \) we write \( Y_i = g(Z_i) \) for a Borel-measurable function \( g \). 
Note that the right-hand side of~\eqref{eq:supp-1} can be written in measure-theoretical terms as
\begin{equation*}
\int g(s)dP_{Z_i(\rho)}(s),
\end{equation*}
where \( P_x \) is the probability measure on states of the program defined by the transition kernel of the process represented by our program in the source state \( x \).

Hence, our aim is to prove that for any \( \sigmaAlg_{i} \)-measureable set \( A \) it holds
\begin{equation}
\label{eq:supp-2}
\int_A g(Z_{i+1}(\rho))d\probm(\rho) = \int_A\left[\int g(s)dP_{Z_i(\rho)}(s) \right] d\probm(\rho),
\end{equation}
where \( \probm \) is the probability measure over the runs of the program. 

We first prove the equality for the case when \( A \) is an \( \sigmaAlg_{i} \)-\emph{cylinder}, i.e. a set of the form \( A = \{\rho \mid Z_1(\rho) \in S_1,\ldots, Z_i(\rho) \in S_i\} \) for some Borel-measurable sets of program states \( S_1,\ldots,S_i \). In such a case, the right-hand side in~\eqref{eq:supp-2} can be rewritten as 
\begin{equation*}
\int_{S_1}\cdots \int_{S_i} g(s_{i+1})dP_{s_i}{(s_{i+1})}\cdots dP_{s_0}(s_1), 
\end{equation*}
which equals the left-hand side of~\eqref{eq:supp-2} directly by the cylinder construction of the probability measure \( \probm \). 

Now to prove that~\eqref{eq:supp-2} holds for any \( \sigmaAlg_{i} \)-measurable set \( A \), assume first that \( g \) is non-negative. By our OST assumption, both integrals in~\eqref{eq:supp-2} are finite. Moreover, the integrals are sigma-additive and an integral over an empty set is zero. Hence, both integrals define a finite measure over \( \sigmaAlg_{i} \) (the measure of set \( A \) being the value of the respective integral when integrating over \( A \)). As shown in the previous paragraph, these two measures agree on the generators of \( \sigmaAlg_{i} \) (the \( \sigmaAlg_{i} \)-cylinders), and the set of these generators is a \( \pi \)-system (is closed under finite intersections). Hence, the two measures are equal on whole \( \sigmaAlg_{i} \) \cite[Lemma 1.6]{Williams91}. Hence, the integrals are the same for all \( A \in \sigmaAlg_{i} \).

For general \( g \), we use the standard trick of splitting \( g \) into the non-negative and negative part: \( g = g^+ - g^{-} \), where both \( g^+ \) and \( g^- \) are non-negative, and hence integrate, on both sides of~\eqref{eq:supp-2}, to the same value as shown in the previous paragraph (and this value is finite by the integrability entailed by OST-soundness), irrespective of the choice of \( A \). The equality of both integrals for \( g \) then follows from the linearity of integrals.

\section{Algorithm for the Similarity Refunation Problem}\label{app:similarityalgo}

We now show how our algorithm can be extended for the Similarity refutation problem.

\paragraph{Additional algorithm parameters.} Recall from Section~\ref{sec:problem} that the Similarity refutation problem is also defined with respect to a metric $d$ over the output space $\mathbb{R}^{|\Vout|}$ and a lower bound $\epsilon > 0$ on the Kantorovich distance that we wish to prove. Our algorithm inputs \( d \) and \( \epsilon \) as parameters. We allow any of the following standard metrics:
\begin{itemize}
	\item {\em $L^p$-metric.} We allow the $L^p$-metric for any natural number $p \in \mathbb{N}$. This encapsulates the standard $L^1$-metric (i.e.~Manhattan metric) and $L^2$-metric (i.e.~Euclidean metric). Given $p \in \mathbb{N}$, the $L^p$-metric $d_p:\mathbb{R}^{|\Vout|} \times \mathbb{R}^{|\Vout|} \rightarrow \mathbb{R}$ is defined via $d_p(x,y) = (\sum_{i=1}^{|\Vout|}(|x[i] - y[i]|)^p)^{1/p}$.
	\item {\em Discrete metric.} We also allow $d$ to be the discrete metric, giving rise to Total Variation distance between output distributions (see Section~\ref{sec:distances}). Recall, the discrete metric $d_0:\mathbb{R}^{|\Vout|} \times \mathbb{R}^{|\Vout|} \rightarrow \mathbb{R}$ is defined via $d_0(x,y) = 0$ if $x=y$ and $d_0(x,y) = 1$ if $x \neq y$.
	\item {\em Uniform metric.} Finally, we allow the $L^{\infty}$-metric (i.e.~uniform metric). The uniform metric $d_\infty:\mathbb{R}^{|\Vout|} \times \mathbb{R}^{|\Vout|} \rightarrow \mathbb{R}$ is defined via $d_\infty(x,y) = \max_{1 \leq i \leq n}|x[y] - y[i]|$.
\end{itemize} 

\paragraph{Algorithm.} Recall from Theorem~\ref{thm:ulesm-refute} that $f$, $U^1_f$ and $L^2_f$ also yield a lower bound on Kantorovich distance between two output distributions, if we in addition constrain $f$ to be {\em $1$-Lipschitz continuous} over reachable output sets $I^1(\locterm^1)$ and $I^2(\locterm^2)$ of the two pCFGs. Hence, our algorithm for the Similarity refutation problem proceeds analogously as in Section~\ref{sec:algoequiv}, with the only difference being that it collects two additional constraints in Step~2 of the algorithm:
\begin{enumerate}
	\item {\em Lower bound on Kantorovich distance.} The algorithm collects the similarity refutation constraint, according to Theorem~\ref{thm:ulesm-refute}:
	\[ L_f(\locinit^2, \vecinit^2) + f((\vecinit^2)^\out - U_f(\locinit^1, \vecinit^1) - f((\vecinit^1)^\out) \geq \epsilon \]
	
	\item {\em $1$-Lipschitz continuity.} Depending on the choice of the metric $d$ over the output space, the algorithm impose the $1$-Lipschitz continuity constraint as follows:
	\begin{itemize}
		\item {\bf\em $1$-Lipschitz continuity w.r.t.~$d_p$.}  Suppose that metric $d$ is the $L^p$-metric $d_p$ for some $p\in\mathbb{N}$. We enforce the following constraint for each pCFG $\pCFG_i$, $i\in\{1,2\}$:
		\begin{equation*}
			\begin{split}
				\forall \mathbf{x}, \mathbf{y}, \mathbf{a} \in \mathbb{R}^{|\Vout|}.\, &\mathbf{x},\mathbf{y} \models I^i(\locterm^i) \land \bigwedge_{j=1}^{|\Vout|} \Big( \mathbf{x}[j] - \mathbf{y}[j] \leq \mathbf{a}[j] \land \mathbf{y}[j] - \mathbf{x}[j] \leq \mathbf{a}[j] \Big) \\
				&\Longrightarrow (f(\mathbf{x}) - f(\mathbf{y}))^p \leq \sum_{j=1}^{|\Vout|} \mathbf{a}[j]^p.
			\end{split}
		\end{equation*}
		The inequality on the right-hand-side is imposed for all $\mathbf{a}[j] \geq |\mathbf{x}[j] - \mathbf{y}[j]|$, which is equivalent to simply imposing it for $\mathbf{a}[j] = |\mathbf{x}[j] - \mathbf{y}[j]|$, giving rise to a sound and complete encoding of the $1$-Lipschitz continuity.
		
		\item {\bf\em $1$-Lipschitz continuity w.r.t.~$d_0$.} The algorithm collects the following constraints on $f$ to be $1$-Lipschitz continuous over outputs of the pCFGs. We enforce the following constraint for each pCFG $\pCFG_i$, $i\in\{1,2\}$:
		\begin{equation*}
			\forall \mathbf{x}, \mathbf{y} \in \mathbb{R}^{|\Vout|}.\, \mathbf{x},\mathbf{y} \models I^i(\locterm^i) \Rightarrow f(\mathbf{x}) - f(\mathbf{y}) \leq 1.
		\end{equation*}
		Indeed, a function $f$ is $1$-Lipschitz continuous if for any distinct $\mathbf{x}, \mathbf{y} \in \mathbb{R}^{|\Vout|}$ the difference in the values of $f$ is at most $1$, since $d_0(\mathbf{x}, \mathbf{y}) = 1$ for $\mathbf{x} \neq \mathbf{y}$. This gives rise to a sound and complete encoding of the $1$-Lipschitz continuity with respect to the discrete metric.
		
		\item {\bf\em $1$-Lipschitz continuity w.r.t.~$d_\infty$.} The algorithm collects the following constraints on $f$ to be $1$-Lipschitz continuous over outputs of the pCFGs. We enforce the following constraint for each pCFG $\pCFG_i$, $i\in\{1,2\}$:
		\begin{equation*}
			\begin{split}
				\forall \mathbf{x}, \mathbf{y}, \mathbf{a} \in \mathbb{R}^{|\Vout|}.\, \forall A \in \mathbb{R}.\, &\mathbf{x},\mathbf{y} \models I^i(\locterm^i) \land \bigwedge_{j=1}^{|\Vout|} \Big( \mathbf{x}[j] - \mathbf{y}[j] \leq \mathbf{a}[j] \land \mathbf{y}[j] - \mathbf{x}[j] \leq \mathbf{a}[j] \Big) \\
				&\bigwedge_{j=1}^{|\Vout|}\Big( \mathbf{a}[j] \leq A \Big) \Longrightarrow f(\mathbf{x}) - f(\mathbf{y}) \leq A.
			\end{split}
		\end{equation*}
		The above constraint is a sound and complete encoding of the fact that $f$ is $1$-Lipschitz continuous over $I^i(\locterm^i)$ for each $i \in \{1,2\}$.
		On the left-hand-side of the entailment, we use the component $\mathbf{a}[j]$ for each $1 \leq j \leq |\Vout|$ to bound from above the absolute difference $|\mathbf{x}[j] - \mathbf{y}[j]|$. Moreover, we use $A$ to bound from above the maximum absolute difference. Thus, the inequality on the right-hand-side is imposed for all $A \geq \mathbf{a}[j] \geq |\mathbf{x}[j] - \mathbf{y}[j]|$, which is equivalent to simply imposing it for $A = \mathbf{a}[j] = |\mathbf{x}[j] - \mathbf{y}[j]|$, giving rise to a sound and complete encoding of the $1$-Lipschitz continuity with respect to the uniform metric.
	\end{itemize}
\end{enumerate}

As in Section~\ref{sec:algoequiv}, the lower-bound constraint is a linear and purely existentially quantified constraint over the symbolic template variables since $\vecinit^1$ and $\vecinit^2$ are fixed. On the other hand, the $1$-Lipschitz continuity constraints are of the same form as in eq.~\eqref{eq:handelmaneq}. Hence, we may proceed analogously as in Steps~3 and~4 in Section~\ref{sec:algoequiv} to translate the collected constraints into an LP instance and reduce the synthesis to LP solving. The algorithm returns ``Not $\epsilon$-output close'' and outputs the computed $f$, $U^1_f$ and $L^2_f$ if the LP is successfully solved, or returns ``Unknown'' otherwise.

The proof of the following theorem can be found in Section~\ref{app:algoproofs} in the supplementary material.

\begin{restatable}[Correctness of Similarity Refutation]{theorem}{algotwosoundness}\label{thm:algotwosoundness}
	Suppose that the algorithm outputs ``Not $\epsilon$-output close''. Then $\pCFG_1$ and $\pCFG_2$ are indeed not $\delta$- output close, and $U^1_f$ and $L^2_f$ are valid UESM and LESM for $f$, respectively.
\end{restatable}

\paragraph{Optimization of the Kantorovich distance.} Finally, we note that our algorithm for the Similarity refutation problem reduces the synthesis of $f$, $U^1_f$ and $L^2_f$ to an LP instance without the optimization objective. Thus, our method can also {\em optimize} the lower bound on the Kantorovich distance by treating \( \epsilon \) as a variable and adding the optimization objective to maximize $\epsilon$.

\section{Soundness Proofs for Algorithms}\label{app:algoproofs}

\algosoundness*

\begin{proof}
	Suppose that the algorithm outputs ``Not output-equivalent'' and that it computes a function $f$ over outputs, a state function $U^1_f$ in $\pCFG_1$ and a state function $L^2_f$ in $\pCFG_2$. We show that $U^1_f$ is an UESM for $f$ in $\pCFG_1$ and that $L^2_f$ is an LESM for $f$ in $\pCFG_2$ which together prove that $\pCFG_1$ and $\pCFG_2$ are not output-equivalent.
	
	Since the algorithm outputs ``Not output-equivalent'', we must have that $f$, $U^1_f$ and $L^2_f$ computed by the algorithm provide a part of the solution to the system of constraints in Step~3. Thus, it follows by  the correctness of reduction to an LP instance that was established in~\cite{ChatterjeeFG16,AsadiC0GM21,Wang0GCQS19,ZikelicCBR22} that they also provide a solution to the system of constraints collected by the algorithm in Step~2. But the constraints collected in Step~2 impose the defining conditions of UESMs as in Definition~\ref{def:uesm}, the defining conditions of LESMs as in Definition~\ref{def:lesm} and OST-soundness conditions as in Definition~\ref{def:ostsound}. Note that the absolute value of the sum of the U/LESM and \( f \) is finite as required by OST soundness, since the U/LESMs and \( f \) are defined via polynomial expressions and all program variables have all moments finite at each step of the program execution. (The latter property follows by a straightforward induction using the fact that the programs use polynomial updates and only sample from distributions that have all moments finite.)
		
	Hence, any solution to the system of constraints collected in Step~2 of the algorithm gives rise to $U^1_f$ and $L^2_f$ that are valid UESM and LESM for $f$, as wanted. Furthermore, by the equivalence refutation constraint that is also collected in Step~2 and by Theorem~\ref{thm:ulesm-refute}, it follows that whenever a solution to the system of constraints in Step~2 exists, the two pCFGs are not output-equivalent. This concludes the proof.
\end{proof}

\algotwosoundness*

\begin{proof}
	Suppose that the algorithm outputs ``Not $\epsilon$-output close'' and that it computes a function $f$ over outputs, a state function $U^1_f$ in $\pCFG_1$ and a state function $L^2_f$ in $\pCFG_2$. We show that $U^1_f$ is an UESM for $f$ in $\pCFG_1$ and that $L^2_f$ is an LESM for $f$ in $\pCFG_2$ which together prove that $\pCFG_1$ and $\pCFG_2$ are not $\epsilon$-output close.
	
	Since the algorithm outputs ``Not $\epsilon$-output close'', we must have that $f$, $U^1_f$ and $L^2_f$ computed by the algorithm provide a part of the solution to the system of constraints in Step~3. Thus, it follows by  the correctness of reduction to an LP instance that was established in~\cite{ChatterjeeFG16,AsadiC0GM21,Wang0GCQS19,ZikelicCBR22} that they also provide a solution to the system of constraints collected by the algorithm in Step~2. But the constraints collected in Step~2 impose the defining conditions of UESMs as in Definition~\ref{def:uesm}, the defining conditions of LESMs as in Definition~\ref{def:lesm} and OST-soundness conditions as in Definition~\ref{def:ostsound}. Hence, any solution to the system of constraints collected in Step~2 of the algorithm gives rise to $U^1_f$ and $L^2_f$ that are valid UESM and LESM for $f$, as wanted. Furthermore, by the similarity refutation constraint that is also collected in Step~2 and by Theorem~\ref{thm:ulesm-refute}, it follows that whenever a solution to the system of constraints in Step~2 exists, the two pCFGs are not $\epsilon$-output close. This concludes the proof.
\end{proof}

\section{Experimental Comparison of OST Conditions}
\label{sec:experiment-comparison-C1-C2-C3}
\begin{table}[h]
	\centering
	\texttt{
		\renewcommand{\arraystretch}{1.2}
	\resizebox{\textwidth}{!}{
		\begin{tabular}{|p{7mm}|c|c c c c|c c c c|c c c c|}
			\hline
			&
			\multirow{2}{*}{Name} & \multicolumn{4}{|c|}{C1/C4} & \multicolumn{4}{|c|}{C2} & \multicolumn{4}{|c|}{C3} \\
			& & Eq. Ref. & T.(s) & Dis. & T.(s) & Eq. Ref & T.(s) & Dis. & T.(s) & Eq. Ref & T.(s) & Dis. & T.(s)\\
			\cline{1-14}
			\parbox[t]{2mm}{\multirow{9}{*}{\rotatebox[origin=c]{90}{\pbox{4cm}{\hfil Benchmarks From \\ \cite{Wang0GCQS19}}}}}
			& {Simple Example} & \checkmark &  0.30 &  266.667 &  0.25 & \checkmark & 1.75 & 0.583 & 32.18 & \checkmark & 0.46 & 266.667 & 0.48 \\
			\cline{2-14}
			& {Nested Loop} & \checkmark &  0.31 &  50.0 &  0.33 & TO & - & TO & - & TO & - & TO & - \\
			\cline{2-14}
			& {Random Walk} & \checkmark &  0.23 &  9.0 &  0.22 & TO & - & TO & - & \checkmark & 1.13 & 11.25 & 0.59 \\
			\cline{2-14}
			& {Goods Discount} & \checkmark &  0.35 &  0.008 &  0.56 & \checkmark & 0.52 & 0.008 & 0.91 & \checkmark & 1.51 & 0.008 & 0.82 \\
			\cline{2-14}
			& {Pollutant Disposal} & \checkmark &  0.44 &  0.026 &  0.46 & TO & - & TO & - & TO & - & TO & - \\
			\cline{2-14}
			& {2D Robot} & \checkmark &  13.90 &  - &  - & TO & - & TO & - & TO & - & TO & - \\
			\cline{2-14}
			& {Bitcoin Mining} & \checkmark &  0.25 &  0.05 &  0.22 &  \checkmark & 0.35 & 0.05 & 0.24 & \checkmark & 0.36 & 0.05 & 0.3 \\
			\cline{2-14}
			& {Bitcoin Mining Pool} & \checkmark &  129.98 &  122761.25 &  131.06 & TO & - & TO & - & \checkmark & 163.00 & TO & - \\
			\cline{2-14}
			& {Species Fight} & \checkmark &  0.90 &  - &  - & TO & - & TO & - & TO & - & TO & - \\
			\hline
			\hline
			\parbox[t]{2mm}{\multirow{7}{*}{\rotatebox[origin=c]{90}{\pbox{4cm}{\hfil Benchmarks From \\ \cite{DBLP:conf/tacas/KuraUH19}}}}}
			&  {coupon\_collector} & \checkmark &  1.10 &  0.5 &  1.41 & TO & - & TO & - & \checkmark & 1.69 & 0.5 & 1.80 \\
			\cline{2-14}
			&  {coupon\_collector4} & \checkmark &  70.96 &  - &  - & TO & - & TO & - & \checkmark & 210.63 & TO & - \\
			\cline{2-14}
			& {random\_walk\_1d\_intvalued} & \checkmark &  0.32 &  1.2 &  0.38 & TO & - & TO & - & \checkmark & 0.98 & 2.4 & 1.85 \\
			\cline{2-14}
			& {random\_walk\_1d realvalued} & \checkmark &  0.27 &  3.841 &  0.50 & TO & - & TO & - & TO & - & TO & - \\
			\cline{2-14}
			& {random\_walk\_1d\_adversary} & \checkmark &  0.38 &  0.768 &  0.40 & TO & - & TO & - & TO & - & TO & - \\
			\cline{2-14}
			& {random\_walk\_2d\_demonic} & \checkmark &  0.58 &  0.668 &  0.58 & TO & - & TO & - & TO & - & TO & - \\
			\cline{2-14}
			& {random\_walk\_2d\_variant} & \checkmark &  0.75 &  0.501 &  0.76 & TO & - & TO & - & TO & - & TO & - \\
			\hline
	\end{tabular}
\renewcommand{\arraystretch}{1}
	}
}
\caption{Comparison of Different OST conditions applied to the first benchmark set}
\label{tab:experiments-C2-C3}
\end{table}

Table \ref{tab:experiments-C2-C3} shows a comparison between performance of different OST conditions when applied for refuting equivalence/similarity of our benchmarks. As expected, (C2) and (C3) are more restrictive, therefore fewer benchmarks could be refuted by them. Moreover, even when with (C2) or (C3) our tool successfully refutes equivalence, they take more time than (C1)/(C4). This shows that although (C2) and (C3) can be applied to a wider range of programs, whenever (C1) or (C4) are applicable, it is more efficient to use the latter.

Note that all of our benchmarks satisfy the assumptions of (C4) while some also satisfy (C1). Moreover, both (C1) and (C4) do not impose any constraints on the generated ESMs, therefore they are presented in the same column in table \ref{tab:experiments-C2-C3}.

\end{document}